\documentclass[a4paper,UKenglish,cleveref,thm-restate]{lipics-v2021}

\usepackage[subrefformat=simple,labelformat=simple]{subcaption}

\usepackage{import}
\usepackage{nicefrac}
\usepackage{mathtools}
\usepackage{algpseudocode}
\usepackage{algorithm}
\usepackage{siunitx}
\usepackage{complexity}
\DeclarePairedDelimiter\ceil{\lceil}{\rceil}

\newcommand{\scaffold}{\ensuremath{\mathcal{T}}}

\newcommand{\supplydemandflow}{\ensuremath{\mathit{f}_{\scaffold}^{\rightarrow}}}

\algdef{SE}[SUBALG]{Indent}{EndIndent}{}{\algorithmicend\ }%
\algtext*{Indent}
\algtext*{EndIndent}

\crefname{claim}{Claim}{Claims}

\hideLIPIcs

\graphicspath{{./figures/}}

\title{Efficiently Reconfiguring a Connected Swarm of Labeled Robots}

\titlerunning{Efficiently Reconfiguring a Connected Swarm of Labeled Robots}

\author{Sándor P. Fekete}{Department of Computer Science, TU Braunschweig, Braunschweig, Germany}{s.fekete@tu-bs.de}{https://orcid.org/0000-0002-9062-4241}{}
\author{Peter Kramer}{Department of Computer Science, TU Braunschweig, Braunschweig, Germany}{p.kramer@tu-bs.de}{https://orcid.org/0000-0001-9635-5890}{}
\author{Christian Rieck}{Department of Computer Science, TU Braunschweig, Braunschweig, Germany}{rieck@ibr.cs.tu-bs.de}{https://orcid.org/0000-0003-0846-5163}{}
\author{Christian Scheffer}{Faculty of Electrical Engineering and Computer Science, Bochum University of Applied Sciences, Bochum, Germany}{christian.scheffer@hs-bochum.de}{https://orcid.org/0000-0002-3471-2706}{}
\author{Arne Schmidt}{Department of Computer Science, TU Braunschweig, Braunschweig, Germany}{arne.schmidt@tu-bs.de}{https://orcid.org/0000-0001-8950-3963}{}

\authorrunning{S.\,P. Fekete, P. Kramer, C. Rieck, C. Scheffer, and A. Schmidt}

\Copyright{Sándor P. Fekete, Peter Kramer, Christian Rieck, Christian Scheffer, Arne Schmidt}



\nolinenumbers

\EventEditors{}
\EventNoEds{0}
\EventLongTitle{}
\EventShortTitle{}
\EventAcronym{}
\EventYear{2022}
\EventDate{}
\EventLocation{}
\EventLogo{}
\SeriesVolume{}
\ArticleNo{A}

\begin{document}

\maketitle

\begin{abstract}
When considering motion planning for a swarm of $n$ labeled robots, we need to rearrange a given start configuration into a desired target configuration via a sequence of parallel, collision-free robot motions.
The objective is to reach the new configuration in a minimum amount of time; an important constraint is to keep the swarm connected at all times.
Problems of this type have been considered before, with recent notable results achieving \emph{constant stretch} for not necessarily connected reconfiguration: If mapping the start configuration to the target configuration requires a maximum Manhattan distance of $d$, the total duration of an overall schedule can be bounded to~$\mathcal{O}(d)$, which is optimal up to constant factors.
However, constant stretch could only be achieved if \emph{disconnected} reconfiguration is allowed, or for scaled configurations (which arise by increasing all dimensions of a given object by the same multiplicative factor) of \emph{unlabeled} robots.

We resolve these major open problems by (1) establishing a lower bound of $\Omega(\sqrt{n})$ for connected, labeled reconfiguration and, most importantly, by (2) proving that for scaled arrangements, constant stretch for connected reconfiguration can be achieved.
In addition, we show that (3) it is \NP-complete to decide whether a makespan of 2 can be achieved, while it is possible to check in polynomial time whether a makespan of 1 can be achieved.
\end{abstract}

\section{Introduction}
\label{sec:introduction}

{
Motion planning for sets of objects is a theoretical and practical problem of  great importance.
A typical task arises from relocating a large collection of agents from a given start into a desired target configuration, while avoiding collisions between objects or with obstacles.
Previous work has largely focused on achieving reconfiguration via sequential schedules, where one agent moves at a time; however, reconfiguring \emph{efficiently} requires reaching the target configuration in a timely or energy-efficient manner, with a natural objective of minimizing the time until completion, called \emph{makespan}.
Problems of labeled reconfiguration play an important role whenever
the involved agents need to be distinguished, such as in automated warehouses and autonomous vehicles;
see Stern et al.~\cite{SternSFK0WLA0KB19} for an overview, and the classic works with further applications
by {\v{S}}vestka and Overmars~\cite{vsvestka1998coordinated} (multi-robot motion planning),
Casal and Yim~\cite{casal1999self} (modular robotics),
and Kornhauser et al.~\cite{kornhauser1984coordinating} (memory management). The algorithmic difficulty
of optimal labeled reconfiguration was established as early as 1979 by Reif~\cite{reif1979complexity},
who gave a proof of \PSPACE-completeness.
Achieving minimum makespan for reconfiguring a swarm of labeled robots was the
subject of the 2021~Computational Geometry Challenge~\cite{challenge2021};
see~\cite{shadoks,gitastrophe,UNIST} for successful contributions.

Exploiting parallelism in a robot swarm to achieve an efficient schedule was
studied in recent seminal work by Demaine et al.~\cite{dfk+-cmprs-19}:
Under certain conditions, a labeled set of robots can be reconfigured with
bounded \emph{stretch}, i.e., there is a collision-free motion plan such that
the makespan of the schedule remains within a constant of the lower bound that
arises from the maximum distance between origin and destination of individual
agents; see also the video by Becker et al.~\cite{coordinated_video} that illustrates these results.

A second important aspect for many applications is \emph{connectivity}~of the
swarm throughout the reconfiguration, because
disconnected pieces may not be able to regain connectivity, and also
because of small-scale swarm robots (such as catoms in claytronics~\cite{goldstein2004claytronics}), which
need connectivity for local motion, electric power and communication; see
the video by Bourgeois~et~al.~\cite{spaceants2}. Connectivity
is not necessarily preserved in the schedules by Demaine~et~al.~\cite{dfk+-cmprs-19}.
In more recent work,
Fekete et al.~\cite{connected-motion-journal} presented an
approach that does achieve constant stretch for \emph{unlabeled} swarms of
robots for the class of \emph{scaled} arrangements; such arrangements
arise by increasing all dimensions of a given object by the same multiplicative
factor and have been considered in previous seminal work on self-assembly,
often with unbounded or logarithmic scale factors (along the lines of
what has been considered in self-assembly~\cite{soloveichik2007complexity}).
The method by Fekete et al.~\cite{connected-motion-journal}
relies strongly on the exchangeability of indistinguishable agents, which
allows a high flexibility in allocating agents to target positions.
However, the labeled setting cannot exploit this flexibility, making it significantly more complex.

These results have left two major open problems.
\begin{enumerate}
    \item Can efficient reconfiguration be achieved in a \emph{connected} manner for a swarm of \emph{labeled} robots in a not necessarily scaled arrangement?
    \item Is it possible to achieve constant stretch for connected reconfiguration of {scaled arrangements} of \emph{labeled} objects?
\end{enumerate}

\subsection{Our contributions}
We resolve both of these open problems.
\begin{enumerate}
    \item We show that connected reconfiguration of a swarm of $n$ labeled robots may require a stretch factor of at least~$\Omega(\sqrt{n})$.
    \item On the positive side, we present an approach to \emph{constant} stretch for connected reconfiguration of scaled arrangements of labeled objects.
    \item In addition, we prove \NP-completeness of deciding whether a makespan of~2 for labeled connected reconfiguration can be achieved.
\end{enumerate}

\subsection{Related work}
Research for multi-agent coordination dates back to the early
years of algorithmic research, such as the seminal
works by Reif~\cite{reif1979complexity} with his proof of
\PSPACE-completeness of optimal labeled reconfiguration, and Schwartz and Sharir~\cite{ss-pmpcbpb-83}
with their results on motion planning for multiple geometric objects.
Efficiently coordinating the motion of many agents 
arises in a large spectrum of applications, such as 
air traffic control~\cite{delahaye2014mathematical},
vehicular traffic networks~\cite{fhtwhfe-mift-11,ss-hbtn-04},
ground swarm robotics~\cite{rubenstein2014programmable,sw-sr-08}, or aerial 
swarm robotics~\cite{cpdsk-saesr-18,kumar}. 
In both discrete and geometric variants of the problem, the objects can be
\emph{labeled}, \emph{colored} or \emph{unlabeled}.  In the \emph{labeled}
case, the objects are all distinguishable and each object has its own, uniquely
defined target position.  In the \emph{colored} case, the objects are
partitioned into $k$ groups and each target position can only be covered by an
object with the right color; see Solovey and Halperin~\cite{sh-kcmrmp-14}.
In the \emph{unlabeled} case, objects are indistinguishable and target positions can be covered by any object;
see Kloder and Hutchinson~\cite{kh-ppimf-06}, Turpin~et~al.~\cite{tmk-tpams-13},
Adler et al.~\cite{adh+-emmpudsp-15}, and Solovey et al.~\cite{syz+-mpudog-15}.
On the negative side, Solovey and Halperin~\cite{sh-hummp-15} prove that the
unlabeled multiple-object motion planning problem is \PSPACE-hard. 
Calinescu, Dumitrescu, and Pach~\cite{CalinescuDP08} consider the sequential reconfiguration of objects lying on vertices of a graph. They give \NP-hardness and inapproximability results for several variants, a 3-approximation algorithm for the unlabeled variant, as well as upper and lower bounds on the number of sequential moves needed.
Geft and Halperin~\cite{GeftH22} show that distance-optimal multi-agent path finding remains \NP-hard on $2$-dimensional grids and multiple empty vertices.
They obtain their result by a linear reduction from \textsc{$3$Sat}, which allows them to exploit the Exponential Time Hypothesis~\cite{ImpagliazzoP01} and therefore to obtain an exponential lower bound on the running time of the problem.

We already mentioned the work by Demaine et al.~\cite{dfk+-cmprs-19} for achieving constant stretch for coordinated motion planning, as well as the recent practical 2021 Computational Geometry  Challenge~\cite{shadoks,challenge2021,gitastrophe,UNIST}.
None of these approaches satisfy the crucial connectivity constraint, which has previously been investigated in terms of decidability and feasibility by Dumitrescu and Pach~\cite{pushing-squares} and Dumitrescu, Suzuki, and Yamashita~\cite{DumitrescuSY04-metamorphic}.
Furthermore, these authors have also proposed efficient patterns for fast swarm locomotion in the plane using sequential moves that allow preservation of connectivity~\cite{DumitrescuSY04-fast-locomotion}.
A closely related body of research concerns itself with sequential pivoting moves that require additional space around moving agents, limiting feasibility and reachability of target states, see publications by Akitaya et al.~\cite{constant-musketeers,compacting-squares}.

Very recently, Fekete et al.~\cite{connected-motion-journal} presented
a number of new results for connected, but unlabeled reconfiguration.
In addion to complexity results for small makespan, they
showed that there is a constant~$c^*$ such that for any pair of start and
target configurations with a (generalized) \emph{scale} of at least~$c^*$, a schedule with
constant stretch can be computed in polynomial time. The involved concept of scale
has received considerable attention in self-assembly; achieving constant scale
has required special cases or operations. Soloveichik and Winfree~\cite{soloveichik2007complexity}
showed that the minimal number of distinct tile types necessary to self-assemble a shape, at
some scale, can be bounded both above and below in terms of the shape’s
Kolmogorov complexity, leading to unbounded scale in general. Demaine et al.~\cite{demaine2011self}
showed that allowing to destroy tiles can be exploited to achieve a scale
that is only bounded by a logarithmic factor, beating the linear bound without such operations.
In a setting of recursive, multi-level \emph{staged} assembly with a logarithmic
number of stages (i.e., ``hands'' for handling subassemblies), Demaine et al.~\cite{ddf-ssnas-08}
achieved logarithmic scale, and constant scale for more constrained classes of polyomino shapes;
this was later improved by Demaine et al.~\cite{dfs+-ngafc-17} to constant scale for a logarithmic
number of stages.
More recently, Luchsinger et al.~\cite{luchsinger2019self} employed repulsive forces between tiles
to achieve constant scale in two-handed self-assembly.

For an extensive overview of multi-agent path finding we refer to~\cite{SternSFK0WLA0KB19}.
Agarwal et al.~\cite{AgarwalGHT23} consider the motion planning problem for unit discs agents in polygonal
domains, with each agent having a \emph{revolving area} around their start and target positions.
They study \emph{weakly-monotone} motion plans, i.e., the agents are ordered and move iteratively with respect to this ordering from their start to their respective target.
To avoid collisions, all other agents are allowed to move only within their revolving area.
They show \APX-hardness for minimizing the total distance traveled, and complementary provide a constant-factor approximation.
Yu and LaValle~\cite{YuL12} discuss the relationship of multi-agent path planning and flow problems in \emph{collision-free unit-distance graphs}.
Charrier et al.~\cite{CharrierQSS19,CharrierQSS19A,CharrierQSS20} study reachability and coverage planning problems for connected agents.
Queffelec, Sankur, and Schwarzentruber~\cite{QueffelecSS23} study the connected multi-agent path finding problem in partially known environments in which the  graph is not known entirely in advance.
Despite the fact that all of this is related to our work, a crucial difference is that we consider the stretch factor as the main performance measure.
}

\subsection{Preliminaries}

We consider \emph{agents} at nodes of the integer infinite grid $G=(V,E)$, where two nodes are connected if and only if they are in unit distance, where distances are measured in the $L_1$~metric.
A~\emph{configuration} is a mapping $C: V\rightarrow \{1,\dots, n,\bot\}$, i.e., each node is mapped injectively to one of the $n$ labeled agents, or to $\bot$ if the node is empty.
The configuration~$C$ is \emph{connected} if the grid graph $H$ that is induced by occupied nodes in~$C$ is connected.
The \emph{silhouette} of a configuration~$C$ is the respective unlabeled configuration, i.e., $C$ without labeling.
Unless stated otherwise, we consider labeled connected configurations.

Two agents are \emph{adjacent} if their positions $v_1,v_2$ are adjacent, i.e.,
$\{v_1,v_2\} \in E(H)$.
An~agent can move in discrete time steps by changing its location from a grid position~$v$ to an adjacent grid position $w$, denoted by $v \rightarrow w$.
Two moves $v_1 \rightarrow w_1$ and $v_2 \rightarrow w_2$ are \emph{collision-free} if $v_1 \neq v_2$ and $w_1 \neq w_2$.
We assume that a \emph{swap} operation, i.e., two moves~${v_1\rightarrow v_2}$ and~${v_2\rightarrow v_1}$, causes a collision and is therefore not allowed in our model.
Note that an agent is allowed to hold its position.
A \emph{transformation} between two configurations $C_1$ and~$C_2$ is a set of collision-free moves $\{ v \rightarrow w \mid C_1(v) = C_2(w)\neq \bot\}$.
For $M \in \mathbb{N}$, a \emph{schedule} with a \emph{makespan} of~$M$ is a sequence $C_1 \rightarrow \cdots \rightarrow C_{M+1}$ (abbreviated as $C_1 \rightrightarrows C_{M+1}$) of transformations. 
A \emph{stable} schedule $C_1 \rightrightarrows_{\chi} C_{M+1}$ uses only connected configurations. 
In the context of this paper, we use these notations equivalently.

Let $C_{s}$ and $C_{t}$ be two connected
configurations with equally many agents called \emph{start} and \emph{target configuration}, respectively.
The \emph{diameter} $d$ of the pair $(C_s,C_t)$ is the maximal Manhattan distance between an agent's start and target position.
The \emph{stretch factor} (or simply \emph{stretch}) of a schedule $C_1 \rightrightarrows_{\chi} C_{M+1}$ is the ratio between its makespan~$M$ and the diameter~$d$ of $(C_s,C_t)$.

\smallskip
With these definitions we can state our problem, called the \textsc{Labeled Connected Coordinated Motion Planning Problem}, as follows. 
Given a pair $(C_s, C_t)$ of labeled connected configurations with equally many agents, and an integer $k$, we are asked to decide whether there is a stable schedule with a makespan of at most $k$ transforming $C_s$ into $C_t$.

\section{Fixed makespan}
\label{sec:fixed-makespan}

Given two labeled configurations, it is easy to see that it can be determined in linear time whether there is a schedule with a makespan of 1 that transforms one into the other:
For every robot, check whether its target position is in distance at most 1;
furthermore, check that no two robots want to swap their positions.
This involves $\mathcal{O}(1)$ checks for every robot, thus $\mathcal{O}(n)$ checks in total.
We obtain the following.

\begin{theorem}
	\label{thm:makespan1}
	It can be decided in $\mathcal{O}(n)$ time whether there is a schedule $C_s \rightrightarrows_\chi C_t$ with makespan $1$ for any pair $(C_s, C_t)$ of labeled connected configurations with $n$ robots.
\end{theorem}

\begin{proof}
	Without loss of generality, we assume that the matching between positions from~$C_s$ and~$C_t$ is perfect.
	Otherwise, a position is occupied multiple times so no feasible reconfiguration is possible.
	We can confirm that $C_s$ and $C_t$ are connected by comparing $n$ to the size of an arbitrary connected component of each.
	The latter can be computed using a graph exploration algorithm such as breadth-first-search in $\mathcal{O}(n)$ time, as the respective dual graph is sparse.

	To check whether there is a schedule with a makespan of $1$, we check for every position~${p\in C_s}$ whether its matched position $p'\in C_t$ is in distance at most $1$; this can be done in~$\mathcal{O}(1)$ time for each position.
	Swaps are forbidden, so we also check if any pair of robots need to swap their positions.
	Because we want to decide whether a schedule with a makespan of~$1$ exists, this only affects pairs of robots that share their neighborhoods.
	Thus, this can be done in~$\mathcal{O}(1)$ time for each position.
	Because a configuration has $n$ vertices, this takes a total of $\mathcal{O}(n)$ time.
\end{proof}

As a direct consequence of~\Cref{thm:makespan1}, we observe the following.

\begin{corollary}\label{con:np}
	Solutions for the \textsc{Labeled Connected Coordinated Motion Planning Problem} can be verified in polynomial time, thus, the problem is in~\NP.
\end{corollary}

On the other hand, Fekete et al.~\cite{connected-motion-journal} showed that it is already \NP-complete to decide whether a schedule with a makespan of 2 can be achieved, if the robot swarm is unlabeled.
\begin{figure}[htb]
	\centering
	\def\svgscale{0.9}
\begingroup%
  \makeatletter%
  \providecommand\color[2][]{%
    \errmessage{(Inkscape) Color is used for the text in Inkscape, but the package 'color.sty' is not loaded}%
    \renewcommand\color[2][]{}%
  }%
  \providecommand\transparent[1]{%
    \errmessage{(Inkscape) Transparency is used (non-zero) for the text in Inkscape, but the package 'transparent.sty' is not loaded}%
    \renewcommand\transparent[1]{}%
  }%
  \providecommand\rotatebox[2]{#2}%
  \newcommand*\fsize{\dimexpr\f@size pt\relax}%
  \newcommand*\lineheight[1]{\fontsize{\fsize}{#1\fsize}\selectfont}%
  \ifx\svgwidth\undefined%
    \setlength{\unitlength}{408.75bp}%
    \ifx\svgscale\undefined%
      \relax%
    \else%
      \setlength{\unitlength}{\unitlength * \real{\svgscale}}%
    \fi%
  \else%
    \setlength{\unitlength}{\svgwidth}%
  \fi%
  \global\let\svgwidth\undefined%
  \global\let\svgscale\undefined%
  \makeatother%
  \begin{picture}(1,1.02752294)%
    \lineheight{1}%
    \setlength\tabcolsep{0pt}%
    \put(0,0){\includegraphics[width=\unitlength,page=1]{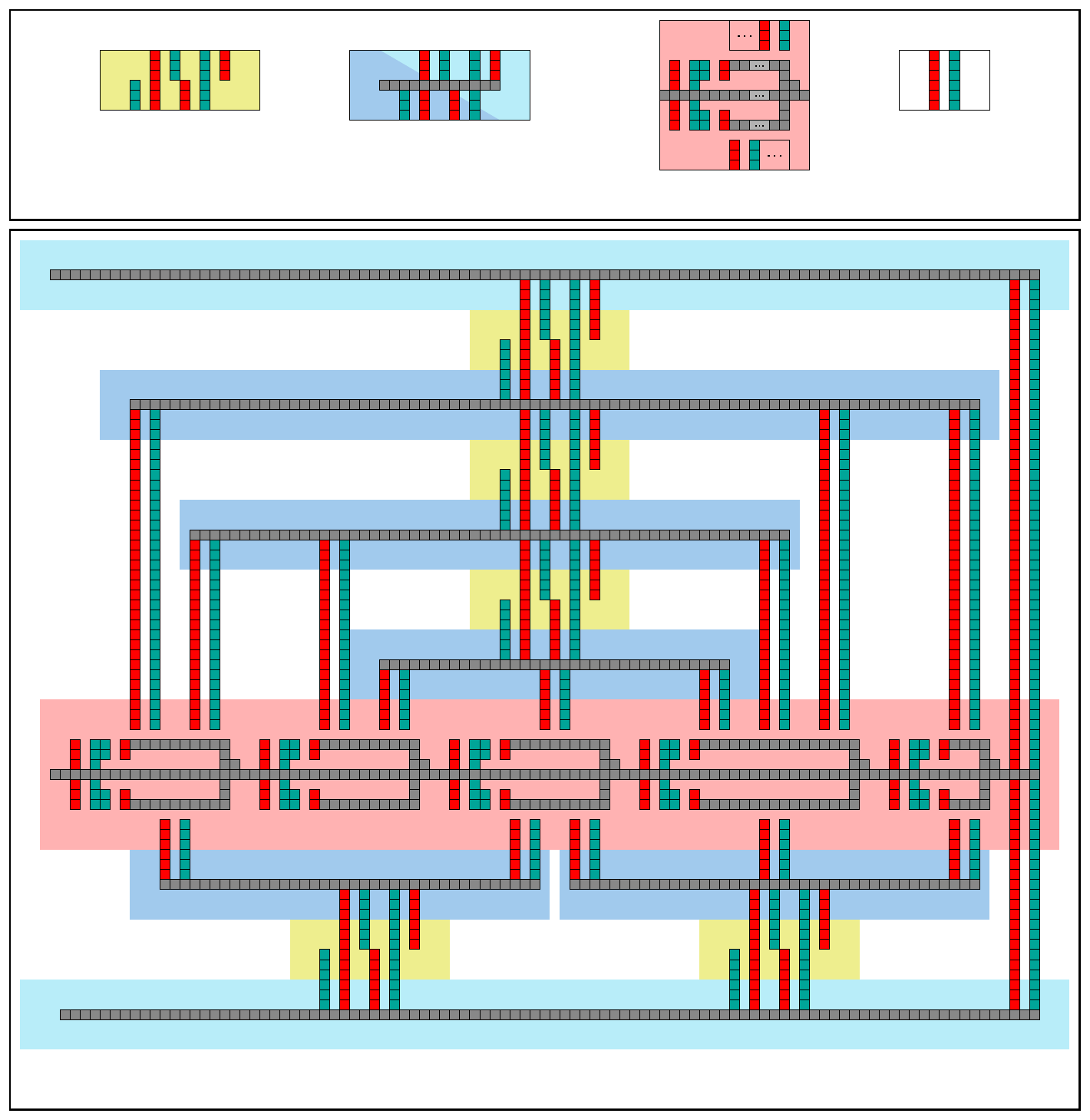}}%
    \put(0.16513762,0.89038128){\color[rgb]{0,0,0}\makebox(0,0)[t]{\lineheight{1.25}\smash{\begin{tabular}[t]{c}separation\\gadget\end{tabular}}}}%
    \put(0.40366974,0.88120696){\color[rgb]{0,0,0}\makebox(0,0)[t]{\lineheight{1.25}\smash{\begin{tabular}[t]{c}clause/auxiliary\\gadget\end{tabular}}}}%
    \put(0.67419261,0.83533541){\color[rgb]{0,0,0}\makebox(0,0)[t]{\lineheight{1.25}\smash{\begin{tabular}[t]{c}variable gadget\end{tabular}}}}%
    \put(0.86726258,0.88990826){\color[rgb]{0,0,0}\makebox(0,0)[t]{\lineheight{1.25}\smash{\begin{tabular}[t]{c}bridges\end{tabular}}}}%
    \put(0.50029017,0.0275229){\color[rgb]{0,0,0}\makebox(0,0)[t]{\lineheight{1.25}\smash{\begin{tabular}[t]{c}The resulting set of configurations constructed for the instance\end{tabular}}}}%
  \end{picture}%
\endgroup%

	\caption{Overview of the structure used for the reduction.}
	\label{fig:overview_hardness}
\end{figure}
Their proof is based on a reduction from the \NP-complete problem \mbox{\textsc{Planar Monotone 3Sat}}, which asks to the decide the satisfiability of a Boolean 3-CNF formula for which the literals in each clause are either all unnegated or all negated, and the corresponding variable-clause incidence graph is planar~\cite{dbk-obspp-10}.
Furthermore, this graph has a planar embedding such that the variables are embedded on a line, and all unnegated clauses are above, while the negated ones are below that variable line.

A brief overview of the used gadgets, alongside a complete construction, is depicted in~\Cref{fig:overview_hardness}.
The high-level idea of their reduction is as follows:
All gadgets are, in both the start and the target configuration, connected via the separation gadgets (and the auxiliary gadget).
Because the robots within each separation gadgets have to move to reach their target position, they cannot maintain connectivity of the whole arrangement on their own, see~\Cref{fig:hardness-separation}.
\begin{figure}[htb]
	\centering
	\def\svgscale{1}
\begingroup%
  \makeatletter%
  \providecommand\color[2][]{%
    \errmessage{(Inkscape) Color is used for the text in Inkscape, but the package 'color.sty' is not loaded}%
    \renewcommand\color[2][]{}%
  }%
  \providecommand\transparent[1]{%
    \errmessage{(Inkscape) Transparency is used (non-zero) for the text in Inkscape, but the package 'transparent.sty' is not loaded}%
    \renewcommand\transparent[1]{}%
  }%
  \providecommand\rotatebox[2]{#2}%
  \newcommand*\fsize{\dimexpr\f@size pt\relax}%
  \newcommand*\lineheight[1]{\fontsize{\fsize}{#1\fsize}\selectfont}%
  \ifx\svgwidth\undefined%
    \setlength{\unitlength}{266.53346745bp}%
    \ifx\svgscale\undefined%
      \relax%
    \else%
      \setlength{\unitlength}{\unitlength * \real{\svgscale}}%
    \fi%
  \else%
    \setlength{\unitlength}{\svgwidth}%
  \fi%
  \global\let\svgwidth\undefined%
  \global\let\svgscale\undefined%
  \makeatother%
  \begin{picture}(1,0.14175882)%
    \lineheight{1}%
    \setlength\tabcolsep{0pt}%
    \put(0,0){\includegraphics[width=\unitlength,page=1]{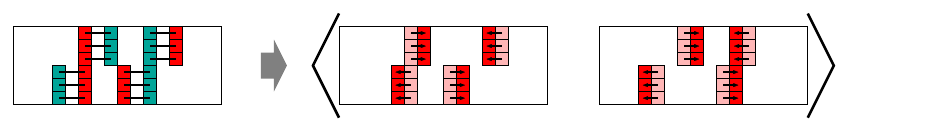}}%
    \put(0.61959102,0.02867092){\color[rgb]{0,0,0}\makebox(0,0)[t]{\lineheight{1.25}\smash{\begin{tabular}[t]{c},\end{tabular}}}}%
  \end{picture}%
\endgroup%

	\caption{The separation gadget.}
	\label{fig:hardness-separation}
\end{figure}
Therefore, robots located in the variable gadgets will have to perform this task in the single intermediate configuration.
These gadgets are designed so that they can expand to the top or bottom, but not both at the same time, as illustrated in~\Cref{fig:hardness-variable}.
\begin{figure}[htb]
	\centering
	\def\svgscale{1}
\begingroup%
  \makeatletter%
  \providecommand\color[2][]{%
    \errmessage{(Inkscape) Color is used for the text in Inkscape, but the package 'color.sty' is not loaded}%
    \renewcommand\color[2][]{}%
  }%
  \providecommand\transparent[1]{%
    \errmessage{(Inkscape) Transparency is used (non-zero) for the text in Inkscape, but the package 'transparent.sty' is not loaded}%
    \renewcommand\transparent[1]{}%
  }%
  \providecommand\rotatebox[2]{#2}%
  \newcommand*\fsize{\dimexpr\f@size pt\relax}%
  \newcommand*\lineheight[1]{\fontsize{\fsize}{#1\fsize}\selectfont}%
  \ifx\svgwidth\undefined%
    \setlength{\unitlength}{266.53344582bp}%
    \ifx\svgscale\undefined%
      \relax%
    \else%
      \setlength{\unitlength}{\unitlength * \real{\svgscale}}%
    \fi%
  \else%
    \setlength{\unitlength}{\svgwidth}%
  \fi%
  \global\let\svgwidth\undefined%
  \global\let\svgscale\undefined%
  \makeatother%
  \begin{picture}(1,0.50756655)%
    \lineheight{1}%
    \setlength\tabcolsep{0pt}%
    \put(0,0){\includegraphics[width=\unitlength,page=1]{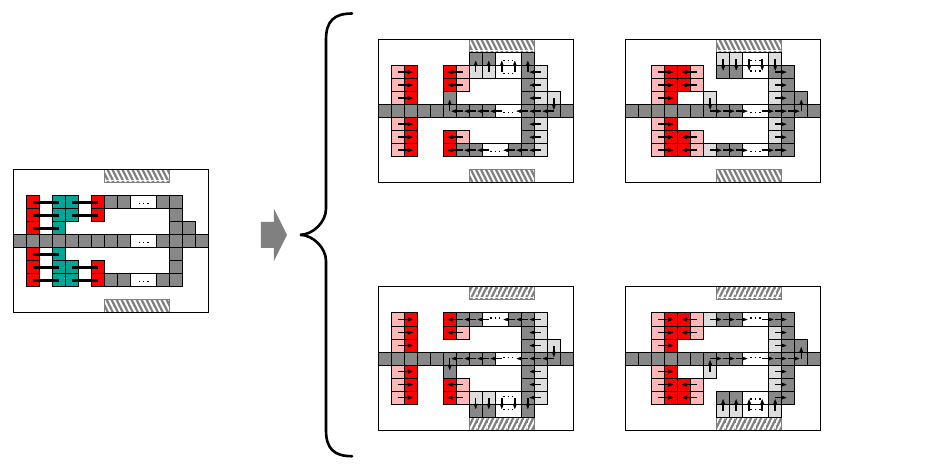}}%
    \put(0.64772999,0.31006114){\color[rgb]{0,0,0}\makebox(0,0)[t]{\lineheight{1.25}\smash{\begin{tabular}[t]{c},\end{tabular}}}}%
    \put(0.64772999,0.04274008){\color[rgb]{0,0,0}\makebox(0,0)[t]{\lineheight{1.25}\smash{\begin{tabular}[t]{c},\end{tabular}}}}%
    \put(0,0){\includegraphics[width=\unitlength,page=2]{variable_gadget_schedule.svg.pdf}}%
    \put(0.95725964,0.38040881){\color[rgb]{0,0,0}\makebox(0,0)[t]{\lineheight{1.25}\smash{\begin{tabular}[t]{c}(1)\end{tabular}}}}%
    \put(0.95725964,0.11308775){\color[rgb]{0,0,0}\makebox(0,0)[t]{\lineheight{1.25}\smash{\begin{tabular}[t]{c}(0)\end{tabular}}}}%
  \end{picture}%
\endgroup%

	\caption{A variable gadget and schedules corresponding to assignments of $0$ and $1$.}
	\label{fig:hardness-variable}
\end{figure}
As the unnegated clauses are at the top, and the negated ones at the bottom, the movement of the variable robots corresponds to an assignment of $1$ or $0$ to the respective variable.

To show \NP-hardness of the \textsc{Labeled Connected Motion Planning Problem}, we will provide a suitable labeling of the configurations and argue that the same construction works for the variant of labeled robot swarms.

\begin{theorem}\label{thm:hardness}
	It is \NP-complete to decide whether there is a schedule $C_s \rightrightarrows_\chi C_t$ with makespan~$2$ for any pair $(C_s, C_t)$ of labeled connected configurations, with $n$ robots.
\end{theorem}

\begin{proof}
We are given an instance $\Pi = (C_s, C_t)$ of the \textsc{Unlabeled Connected Motion Planning Problem}~that arises via the polynomial-time reduction from an instance of \textsc{Planar Monotone 3Sat}.
We construct an instance $\Pi_\ell = (C_s^\ell, C_t^\ell)$ of \textsc{Labeled Connected Motion Planning} by labeling the positions of the start and the target configuration, such that a schedule with makespan 2 for $\Pi_\ell$ is also a schedule for $\Pi$.
	
We create the labeling as follows, using three different colors to indicate occupied positions in the start configuration (red), in the target configuration (dark cyan), and in both configurations~(gray). 
First, the labeling for every gray position is identical in both configurations, $C_s^\ell$ and $C_t^\ell$. 
A makespan of $2$ confines each robot's movement to an area of small radius.
Projecting this radius onto any red position in the clause/auxiliary gadgets, the bridges and the variable gadgets of $\Pi$, leaves only a single dark cyan position in range.
These red-cyan pairs get the same labeling in $C_s^\ell$ and $C_t^\ell$, respectively. 
In the separation gadget there is a red position with two possible dark cyan positions in range, see~\Cref{fig:separation-match}. 
\begin{figure}[t]
	\centering
	\begin{subfigure}{0.5\textwidth}%
		\centering
		\def\svgscale{1.8}
		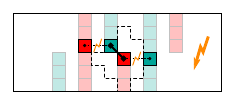
		\caption{}
		\label{fig:separation_gadget_mismatch-a}
	\end{subfigure}%
	\begin{subfigure}{0.5\textwidth}
		\centering
		\def\svgscale{1.8}
		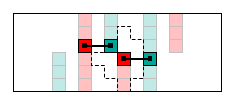
		\caption{}
		\label{fig:separation_gadget_mismatch-b}
	\end{subfigure}%
	\caption{For the separation gadget, there is only one feasible matching. The dotted lines indicate the range of motion for the central red robot.}
	\label{fig:separation-match}
\end{figure}
However, only one of these admits a feasible matching of all positions, i.e., the labeling is also unique in these separation gadgets.
Due to the unique labeling within each gadget, it is straightforward to observe that there is a schedule with a makespan of $2$ (implying that there are two different schedules for the variable gadget due to the two possible assignments) realizing the reconfiguration given by the labeling. 
Because the labeling is unique, the schedule for $\Pi_\ell$ yields a schedule for~$\Pi$.
Hence, we conclude that the labeled variant is as least as hard as the unlabeled one.  
	
Together with~\Cref{con:np}, this completes the proof.
\end{proof}

As a consequence, even approximating the makespan is \NP-hard: If no schedule of makespan \num{2} is found,
then the makespan is at least \num{3}.

\begin{corollary}\label{cor:connected-motion-planning-optmial-hard}
	It is \NP-hard to compute for a pair of labeled connected configurations $C_s$ and~$C_t$ with $n$ robots, a stable schedule that transforms $C_s$ into $C_t$ within a constant of $(\nicefrac{3}{2}-\varepsilon)$ (for any $\varepsilon>0$) of the minimum makespan.
\end{corollary}

\section{Lower bound on stretch factor}
\label{sec:lower-bound-on-stretch-factor}

We show a lower bound of $\Omega(\sqrt{n})$ on the stretch factor.
For this, we consider the pair of labeled~connected configurations $(C_s,C_t)$ shown in~\Cref{fig:lower-bound-a}, both consisting of $n$ robots.
The difference between both configurations is that adjacent robots need to swap their positions.
Thus, the diameter of $(C_s,C_t)$ is $d = 1$.
Because swaps are not allowed within the underlying model, some robots have to move orthogonally.
It is easy to see that a single transformation allows for at most two robots to move orthogonally without disconnecting the configuration.
A crucial insight is that each of these robots can realize at most one swap~in~parallel.
\begin{figure}[htb]
	\centering
	\begin{subfigure}[b]{0.6\linewidth}
		\centering
		\def\svgscale{1.8}
\begingroup%
  \makeatletter%
  \providecommand\color[2][]{%
    \errmessage{(Inkscape) Color is used for the text in Inkscape, but the package 'color.sty' is not loaded}%
    \renewcommand\color[2][]{}%
  }%
  \providecommand\transparent[1]{%
    \errmessage{(Inkscape) Transparency is used (non-zero) for the text in Inkscape, but the package 'transparent.sty' is not loaded}%
    \renewcommand\transparent[1]{}%
  }%
  \providecommand\rotatebox[2]{#2}%
  \newcommand*\fsize{\dimexpr\f@size pt\relax}%
  \newcommand*\lineheight[1]{\fontsize{\fsize}{#1\fsize}\selectfont}%
  \ifx\svgwidth\undefined%
    \setlength{\unitlength}{97.5bp}%
    \ifx\svgscale\undefined%
      \relax%
    \else%
      \setlength{\unitlength}{\unitlength * \real{\svgscale}}%
    \fi%
  \else%
    \setlength{\unitlength}{\svgwidth}%
  \fi%
  \global\let\svgwidth\undefined%
  \global\let\svgscale\undefined%
  \makeatother%
  \begin{picture}(1,0.2)%
    \lineheight{1}%
    \setlength\tabcolsep{0pt}%
    \put(0.00145373,0.02931562){\color[rgb]{0,0,0}\makebox(0,0)[t]{\lineheight{1.25}\smash{\begin{tabular}[t]{c}$C_t$\end{tabular}}}}%
    \put(0,0){\includegraphics[width=\unitlength,page=1]{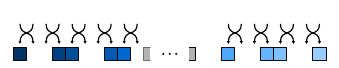}}%
    \put(0.00145373,0.1447002){\color[rgb]{0,0,0}\makebox(0,0)[t]{\lineheight{1.25}\smash{\begin{tabular}[t]{c}$C_s$\end{tabular}}}}%
    \put(0,0){\includegraphics[width=\unitlength,page=2]{counterexample_description_a.svg.pdf}}%
  \end{picture}%
\endgroup%

		\caption{}
		\label{fig:lower-bound-a}
	\end{subfigure}\hfil
	\begin{subfigure}[b]{0.4\linewidth}
		\centering
		\def\svgscale{1.8}
		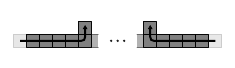
		\caption{}
		\label{fig:lower-bound-b}
	\end{subfigure}\hfil
	\caption{Pairs of robots must swap their positions (a), which can only be realized by schedules using moves that involve all robots (b).}
	\label{fig:lower-bound}
\end{figure}

\begin{theorem}
	There are pairs of labeled connected configurations $C_s$ and $C_t$ with $n$ robots, so that every schedule $C_s \rightrightarrows_\chi C_t$ has a stretch factor of at least $\Omega(\sqrt{n})$.
	\label{the:sqrt-n-stretch}
\end{theorem}

\begin{proof}
	Consider a start configuration consisting of $n$ robots arranged in a straight line, and a target configuration that arises by swapping each adjacent pair of robots.
	This pair of configurations has a diameter of $d=1$, with a total of $\nicefrac{n}{2}$ swaps.
	A~minimum makespan may be achieved by rearranging adjacent robots in parallel.
	This parallelism is limited by the rate at which robots can move orthogonally to the linear arrangement, see~\Cref{fig:lower-bound-b}.
	
	For some $\lambda\in\{0,\dots,\lceil\nicefrac{n}{2}\rceil\}$, the process of freeing $2\lambda$ robots out of the line such that they can navigate along the configuration freely would take a total of $\lambda $ transformation steps.
	With each free robot we can reduce the number of swaps of the remaining line by $1$ in a constant number of steps.
	Therefore, realizing all swaps takes at least $\Omega(\nicefrac{n}{\lambda})$ transformation steps, resulting in a total makespan of $\Omega(\lambda + \nicefrac{n}{\lambda})$.
	This is asymptotically minimal for $\lambda=\sqrt{n}$, resulting in a makespan of at least $\Omega(\sqrt{n})$.
	Because of the diameter $d=1$, this implies that any schedule has stretch at least $\Omega(\sqrt{n})$.
\end{proof}

\section{Preparing global movement locally}
\label{sec:localreconfiguration}
In this section we introduce a preliminary result that hinges on exploiting a global support structure that provides us with powerful local reconfiguration routines.

Our problem setting of \textsc{Labeled Connected Coordinated Motion Planning} naturally permits efficient parallelization by decomposition of the underlying grid into disjoint regions of comparable size, enabling us to perform operations in all of them simultaneously.

As an important constraint, we observe that these regions should be strongly connected -- in most scenarios, a robot must be able to move efficiently from one region to another in order to reach its target position.
Therefore, we seek to find a partition that allows for both high parallelization and high interconnectivity between these disjoint regions.
This concept was previously exploited for in-place reconfiguration of rectangular arrangements of robots by Demaine et al.~\cite{dfk+-cmprs-19}, who combined a grid-based partition and the so-called \textsc{RotateSort} algorithm for efficient reconfiguration.
To this end, they define a grid that scales with the diameter $d$ of the input configuration to partition the configuration into disjoint \emph{tiles}, each of which is an $\mathcal{O}(d)\times\mathcal{O}(d)$ square of robots.
The resulting partition of the grid space is what we call an \emph{$\mathcal{O}(d)$-tiling}, as shown in~\Cref{fig:tiling-example}.

\begin{figure}[ht]
	\begin{subfigure}[b]{74.25bp}
		\def\svgscale{0.9}
		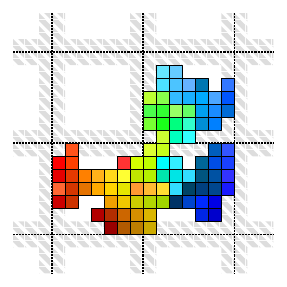
		\caption{}
		\label{fig:tiling-example}
	\end{subfigure}\hfill%
	\begin{subfigure}[b]{74.25bp}
		\def\svgscale{0.9}
		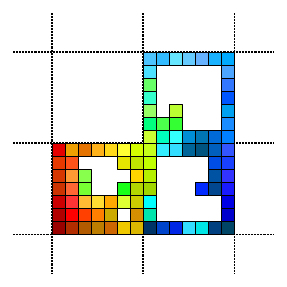
		\caption{}
		\label{fig:tiled-config-example}
	\end{subfigure}\hfill%
	\begin{subfigure}[b]{148.5bp} 
		\def\svgscale{0.9}%
		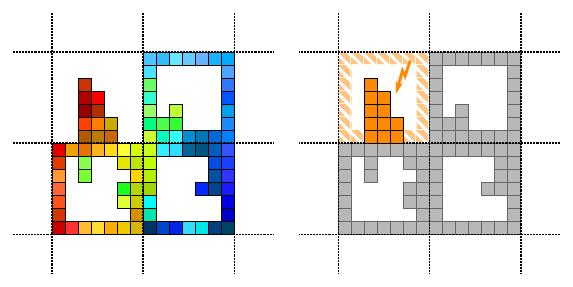
		\caption{}
		\label{fig:nontiled-config-example}
	\end{subfigure}
	\caption{
		Tilings and tiled configurations.
		In (a), we see a configuration, a $7$-tiling, and its boundary regions.
		A valid tiled configuration can be seen in (b), alongside an invalid one in (c) with a missing tile boundary  highlighted in orange.
	}
	\label{fig:tiling-and-scaffold}
\end{figure}
\begin{figure}[htb]
	\centering
	\begin{subfigure}[b]{124.875bp}
		\def\svgscale{0.9}
\begingroup%
  \makeatletter%
  \providecommand\color[2][]{%
    \errmessage{(Inkscape) Color is used for the text in Inkscape, but the package 'color.sty' is not loaded}%
    \renewcommand\color[2][]{}%
  }%
  \providecommand\transparent[1]{%
    \errmessage{(Inkscape) Transparency is used (non-zero) for the text in Inkscape, but the package 'transparent.sty' is not loaded}%
    \renewcommand\transparent[1]{}%
  }%
  \providecommand\rotatebox[2]{#2}%
  \newcommand*\fsize{\dimexpr\f@size pt\relax}%
  \newcommand*\lineheight[1]{\fontsize{\fsize}{#1\fsize}\selectfont}%
  \ifx\svgwidth\undefined%
    \setlength{\unitlength}{138.75bp}%
    \ifx\svgscale\undefined%
      \relax%
    \else%
      \setlength{\unitlength}{\unitlength * \real{\svgscale}}%
    \fi%
  \else%
    \setlength{\unitlength}{\svgwidth}%
  \fi%
  \global\let\svgwidth\undefined%
  \global\let\svgscale\undefined%
  \makeatother%
  \begin{picture}(1,0.48648649)%
    \lineheight{1}%
    \setlength\tabcolsep{0pt}%
    \put(0,0){\includegraphics[width=\unitlength,page=1]{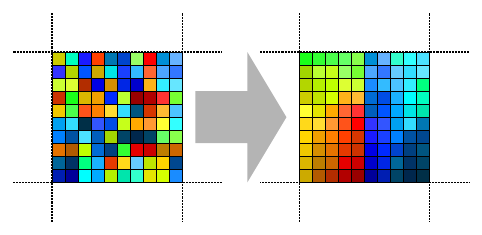}}%
    \put(0.4947786,0.21869038){\makebox(0,0)[t]{\lineheight{1.25}\smash{\begin{tabular}[t]{c}$\mathcal{O}(d)$\end{tabular}}}}%
  \end{picture}%
\endgroup%

		\caption{}
		\label{fig:rotatesort}
	\end{subfigure}\hfill%
	\begin{subfigure}[b]{229.5bp}
		\def\svgscale{0.9}
\begingroup%
  \makeatletter%
  \providecommand\color[2][]{%
    \errmessage{(Inkscape) Color is used for the text in Inkscape, but the package 'color.sty' is not loaded}%
    \renewcommand\color[2][]{}%
  }%
  \providecommand\transparent[1]{%
    \errmessage{(Inkscape) Transparency is used (non-zero) for the text in Inkscape, but the package 'transparent.sty' is not loaded}%
    \renewcommand\transparent[1]{}%
  }%
  \providecommand\rotatebox[2]{#2}%
  \newcommand*\fsize{\dimexpr\f@size pt\relax}%
  \newcommand*\lineheight[1]{\fontsize{\fsize}{#1\fsize}\selectfont}%
  \ifx\svgwidth\undefined%
    \setlength{\unitlength}{255bp}%
    \ifx\svgscale\undefined%
      \relax%
    \else%
      \setlength{\unitlength}{\unitlength * \real{\svgscale}}%
    \fi%
  \else%
    \setlength{\unitlength}{\svgwidth}%
  \fi%
  \global\let\svgwidth\undefined%
  \global\let\svgscale\undefined%
  \makeatother%
  \begin{picture}(1,0.32352941)%
    \lineheight{1}%
    \setlength\tabcolsep{0pt}%
    \put(0,0){\includegraphics[width=\unitlength,page=1]{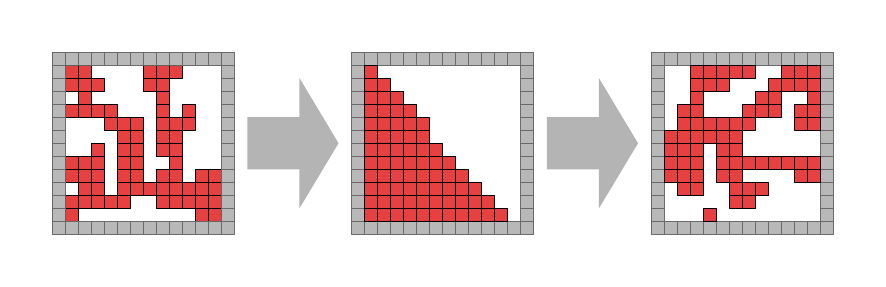}}%
    \put(0.32804121,0.14840501){\makebox(0,0)[t]{\lineheight{1.25}\smash{\begin{tabular}[t]{c}$\mathcal{O}(d)$\end{tabular}}}}%
    \put(0.6659467,0.14969359){\makebox(0,0)[t]{\lineheight{1.25}\smash{\begin{tabular}[t]{c}$\mathcal{O}(d)$\end{tabular}}}}%
    \put(0,0){\includegraphics[width=\unitlength,page=2]{interior_triangle_unlabeled.svg.pdf}}%
  \end{picture}%
\endgroup%

		\caption{}
		\label{fig:interior-triangle-example}
	\end{subfigure}\hfil
	\caption{
		We exploit two preliminary results: \textsc{RotateSort}, i.e., \Cref{the:rotatesort} can be applied to sort rectangular regions such as in (a).
		For unlabeled tiled configurations, \Cref{lem:unlabelled_interior_od}
		allows us to modify the interior regions, as shown in (b).
	}
	\label{fig:two_subroutines}
\end{figure}
We note their use of \textsc{RotateSort} within these tiles, as we will later apply it to a similar effect.
See~\Cref{fig:rotatesort} for an illustration of its capabilities.
\begin{lemma}[\textsc{RotateSort}]
	\label{the:rotatesort}
	Let $C_s$ and $C_t$ be two labeled configurations of an $n_1 \times n_2$ rectangle with $n_1 > n_2 \geq 2$.
	We can compute in polynomial time a schedule $C_s\rightrightarrows_\chi C_t$ with makespan $\mathcal{O}(n_1+n_2)$.
\end{lemma}

\medskip
Prior research on unlabeled robots under connectivity constraints employs a support structure along tile-based subdivisions of the plane to maintain connectivity.
As our work involves a similar support structure, we define the following terms.
The \emph{boundary} of a tile is composed of the grid positions directly adjacent to the outer edge.
Analogously, the \emph{interior} positions are the remaining positions within the tile region.
The \emph{scaffold} of a tiling is the union of all boundaries of its non-empty tiles.
As these are the positions that we want our support structure to occupy, we say that a \emph{tiled configuration} is a connected configuration whose silhouette is a subset of the given tiling and a superset of its scaffold.
This concept is visualized in~\Cref{fig:tiled-config-example,fig:nontiled-config-example}.

Fekete et al.~\cite{connected-motion-journal} obtained the following result
for \emph{unlabeled} reconfiguration under connectivity constraints.
\begin{lemma}
	\label{lem:unlabelled_interior_od}
	Let $C_s$ and $C_t$ be two unlabeled $\mathcal{O}(d)$-tiled configurations such that $C_s$ and~$C_t$ contain the same number of robots in the interior of each tile $T$.
	We can compute in polynomial time a schedule $C_s\rightrightarrows_\chi C_t$ with makespan $\mathcal{O}(d)$.
\end{lemma}
This result employs the reachability of a canonical intermediate structure within tiles i.e, an intermediate configuration derived only from the number of robots within each tile.
We will briefly reference this in the following section.
Intuitively, this intermediate configuration is obtained by a monotone, triangle-based compaction of the interior robots into one of the tile's corners, as illustrated in~\Cref{fig:interior-triangle-example}.

\medskip
Building on these two preliminary results, we now construct a subroutine that can be used to transform all tiles of a tiled configuration within a makespan of $\mathcal{O}(d)$.
\begin{theorem}
	\label{the:local_od}
	For any two labeled tiled configurations $C_s$ and~$C_t$ for which each tile consists of the same robots, we can compute in polynomial time a stable schedule of makespan $\mathcal{O}(d)$ that transforms one into the other.
\end{theorem}

\medskip

In the remainder of this subsection we provide three results which, in combination, yield a proof of \Cref{the:local_od}.
To this end, we first show that the interior of a tile can be reconfigured arbitrarily, followed by a description of how this can be adapted to include the reconfiguration of the boundary, as well as arbitrary exchanges of robots between a tile's interior and its~boundary.

\subsection{Reconfiguring the interior}
We start by showing how the interior of a tile can be transformed arbitrarily.
As moves are reversible, proving that a canonical configuration is reachable is sufficient.
In our case, this configuration will be a compact subset of a square.
In order to obtain such a configuration, we first apply a direct result of~\Cref{lem:unlabelled_interior_od}, which is that we can modify the silhouette of a tile's interior configuration arbitrarily in $\mathcal{O}(d)$.
Rather than stopping at the previously described canonical configuration based on triangular shapes, we deterministically create a configuration with a square-like silhouette.
We arrive at the following corollary for $k$~interior~robots.
\begin{corollary}
	\label{cor:underfull_square_creation}
	For any two connected configurations consisting of an immobile $m \times m$ boundary for $m\in \mathcal{O}(d)$ and $k\leq (m-2)^2$ interior robots, there is a stable schedule of makespan $\mathcal{O}(d)$ that moves all interior robots into a subset of
	$\ceil{\sqrt{k}}\times \ceil{\sqrt{k}}$ positions.
\end{corollary}

\medskip
This process is visualized in~\Cref{fig:interior-reconfiguration-overview}.
Using the structural results from~\Cref{lem:unlabelled_interior_od}, we transform an initial interior configuration $C_s$ into a square-like silhouette $C_\Box$.
Due to the relationship of any two sequential square numbers, i.e., $x^2-(x-1)^2=2x-1$, at most $2\ceil{\sqrt{k}} - 1$ nodes within the $\ceil{\sqrt{k}}\times \ceil{\sqrt{k}}$ square may end up unoccupied.
We form the silhouette in a manner such that these robots are missing within the top two rows, from an arbitrary but fixed end.
In a final step concerning the interior, we now argue that we can efficiently exchange the interior robot's positions in the silhouette, as indicated in~\Cref{fig:interior-reconfiguration-overview}.
\begin{figure}[htb]
	\centering
	\def\svgscale{0.9}
\begingroup%
  \makeatletter%
  \providecommand\color[2][]{%
    \errmessage{(Inkscape) Color is used for the text in Inkscape, but the package 'color.sty' is not loaded}%
    \renewcommand\color[2][]{}%
  }%
  \providecommand\transparent[1]{%
    \errmessage{(Inkscape) Transparency is used (non-zero) for the text in Inkscape, but the package 'transparent.sty' is not loaded}%
    \renewcommand\transparent[1]{}%
  }%
  \providecommand\rotatebox[2]{#2}%
  \newcommand*\fsize{\dimexpr\f@size pt\relax}%
  \newcommand*\lineheight[1]{\fontsize{\fsize}{#1\fsize}\selectfont}%
  \ifx\svgwidth\undefined%
    \setlength{\unitlength}{408.75bp}%
    \ifx\svgscale\undefined%
      \relax%
    \else%
      \setlength{\unitlength}{\unitlength * \real{\svgscale}}%
    \fi%
  \else%
    \setlength{\unitlength}{\svgwidth}%
  \fi%
  \global\let\svgwidth\undefined%
  \global\let\svgscale\undefined%
  \makeatother%
  \begin{picture}(1,0.23853211)%
    \lineheight{1}%
    \setlength\tabcolsep{0pt}%
    \put(0,0){\includegraphics[width=\unitlength,page=1]{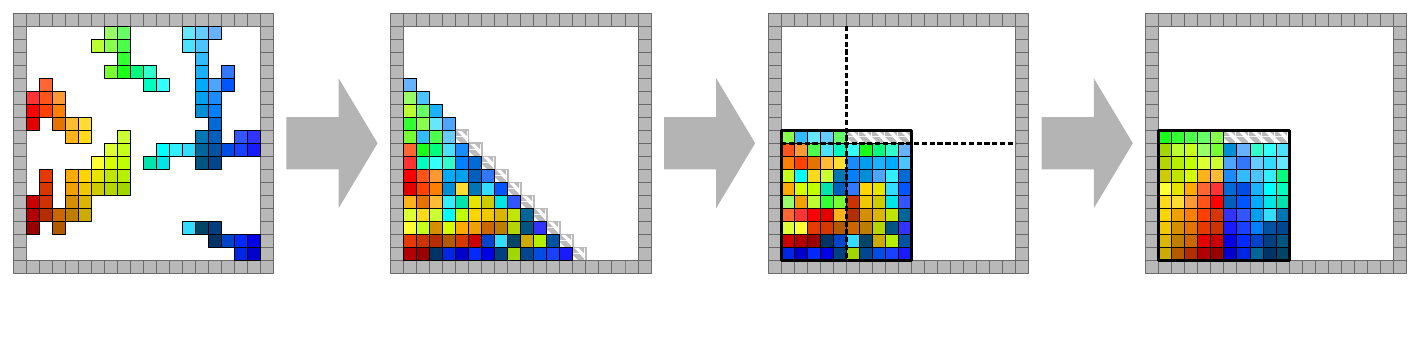}}%
    \put(0.49724763,0.12844038){\makebox(0,0)[t]{\lineheight{1.25}\smash{\begin{tabular}[t]{c}$\mathcal{O}(d)$\end{tabular}}}}%
    \put(0.76330265,0.12844038){\makebox(0,0)[t]{\lineheight{1.25}\smash{\begin{tabular}[t]{c}$\mathcal{O}(d)$\end{tabular}}}}%
    \put(0.10091743,0.01834862){\makebox(0,0)[t]{\lineheight{1.25}\smash{\begin{tabular}[t]{c}$C_s$\end{tabular}}}}%
    \put(0.63302747,0.01834864){\makebox(0,0)[t]{\lineheight{1.25}\smash{\begin{tabular}[t]{c}$C_\Box$\end{tabular}}}}%
    \put(0.89908275,0.01834864){\makebox(0,0)[t]{\lineheight{1.25}\smash{\begin{tabular}[t]{c}$C_\Box'$\end{tabular}}}}%
    \put(0.2311926,0.12844038){\makebox(0,0)[t]{\lineheight{1.25}\smash{\begin{tabular}[t]{c}$\mathcal{O}(d)$\end{tabular}}}}%
  \end{picture}%
\endgroup%

	\caption{We obtain the canonical square-like silhouette $C_\Box$ in a tile's interior using the canonical triangle structure from \Cref{lem:unlabelled_interior_od}, which we can then reorder arbitrarily.}
	\label{fig:interior-reconfiguration-overview}
\end{figure}

\begin{lemma}
	The interiors of all tiles in a labeled tiled configuration that consists of $m\times m$-tiles with $m\in \mathcal{O}(d)$ may be locally rearranged by a schedule of makespan $\mathcal{O}(d)$.
	\label{lem:labeled_interior_od}
\end{lemma}

\begin{proof}
	We prove this by means of a process that can be applied to all tiles in parallel.
	Consider an arbitrary tile that has $k\in [6,(m-2)^2]$ interior robots.
	We first employ~\Cref{cor:underfull_square_creation} to create the square-like interior silhouette $C_\Box$.
	For $k$ robots, this is a subset of a $\ceil{\sqrt{k}}\times \ceil{\sqrt{k}}$ square.
	There now exist two cases as follows.

	If the interior silhouette is a fully occupied rectangle, we can apply \Cref{the:rotatesort} to this rectangle to rearrange the robots in $\mathcal{O}(d)$.

	Otherwise, we cover the occupied area by two rectangles, one of maximal width in the ``lower'' portion, and one of maximal height in the ``taller'' segment, as indicated by the cuts through $C_\Box$ in \Cref{fig:interior-reconfiguration-overview}.
	We can then reorder the robots by applying \Cref{the:rotatesort} multiple times as follows.
	First, we ensure that the robots for the incomplete top row are located in the base of the taller rectangle, by reordering the lower rectangle.
	Then, a second application to the enclosing section ensures that the incomplete row is correctly configured.
	For the case of a $1$-wide rectangle, the empty position next to the robot may be considered a ``pseudo-robot'' and
	the operation may be applied to the corresponding $2$-wide area instead.
	Finally, a third repetition ensures that the lower rectangle is also correctly ordered.
	Overall, this yields a makespan of $\mathcal{O}(d)$.

	By appropriately reordering the robots in the silhouette of $C_\Box$ as above, before applying \Cref{cor:underfull_square_creation} in reverse, we have the means to create an arbitrary target configuration of the given tile's interior.
	Note that the necessary ordering of robots can be computed by a simple matching of positions in the silhouette of $C_\Box$ with respect to the start- and target arrangements.
	This concludes our proof of~\Cref{lem:labeled_interior_od}.
\end{proof}
Note that the case of $k \in [1,5]$ is excluded from the proof, but trivial schedules utilizing the existence of the boundary may be determined, allowing for arbitrary rearrangement of~the~interior.

\subsection{Reconfiguring the boundary}
\label{subsubsec:scaffold_reordering_local}
We now show that the scaffolding structure can be locally modified in-place, i.e., without changing its silhouette at any point of the process.

\begin{lemma}
	The boundaries of all tiles in a labeled tiled configuration that consists of ${m\times m}$-tiles with $m\in \mathcal{O}(d)$ may be locally reordered by a schedule of makespan $\mathcal{O}(d)$.
	\label{lem:parallel_boundary_reordering}
\end{lemma}

\begin{proof}
	We must allow adjacent tiles to take turns in utilizing the shared portion of the boundary as a processing area for themselves.
	As the dual graph of the tiling is bipartite, the tiles can be partitioned into two disjoint sets, each of which can be reconfigured in parallel.
	On a high level, we achieve this by repeatedly gathering and sorting $m-2$ robots at a time in this shared section of the boundary, incrementally arranging the scaffold into its~target~state.

	Consider now a single tile that is permitted to employ part of one neighbor's boundary.
	We can pick an arbitrary continuous sequence of $m-2$ robots from the target configuration and start gathering the corresponding robots in the ``borrowed'' section of the neighbor tile's scaffold.
	To this end, we repeatedly shift the robots along the boundary by $m-2$ units, before applying~\Cref{the:rotatesort} to switch desired robots to their target position in the sorted sequence in $\mathcal{O}(m)$ transformations.
	After at most five iterations of this process, we have obtained the desired sorted sequence, which we can then reinsert to the boundary of our tile at an appropriate location.

	This process must be repeated up to four more times to obtain an arbitrarily reodered boundary configuration and restore the neighbor's boundary to its original state.
	For an illustration of a single iteration, see~\Cref{fig:boundary_circulation}.
\end{proof}
\begin{figure}[ht]
	\centering
	\def\svgscale{0.9}
\begingroup%
  \makeatletter%
  \providecommand\color[2][]{%
    \errmessage{(Inkscape) Color is used for the text in Inkscape, but the package 'color.sty' is not loaded}%
    \renewcommand\color[2][]{}%
  }%
  \providecommand\transparent[1]{%
    \errmessage{(Inkscape) Transparency is used (non-zero) for the text in Inkscape, but the package 'transparent.sty' is not loaded}%
    \renewcommand\transparent[1]{}%
  }%
  \providecommand\rotatebox[2]{#2}%
  \newcommand*\fsize{\dimexpr\f@size pt\relax}%
  \newcommand*\lineheight[1]{\fontsize{\fsize}{#1\fsize}\selectfont}%
  \ifx\svgwidth\undefined%
    \setlength{\unitlength}{408.75bp}%
    \ifx\svgscale\undefined%
      \relax%
    \else%
      \setlength{\unitlength}{\unitlength * \real{\svgscale}}%
    \fi%
  \else%
    \setlength{\unitlength}{\svgwidth}%
  \fi%
  \global\let\svgwidth\undefined%
  \global\let\svgscale\undefined%
  \makeatother%
  \begin{picture}(1,0.2293578)%
    \lineheight{1}%
    \setlength\tabcolsep{0pt}%
    \put(0,0){\includegraphics[width=\unitlength,page=1]{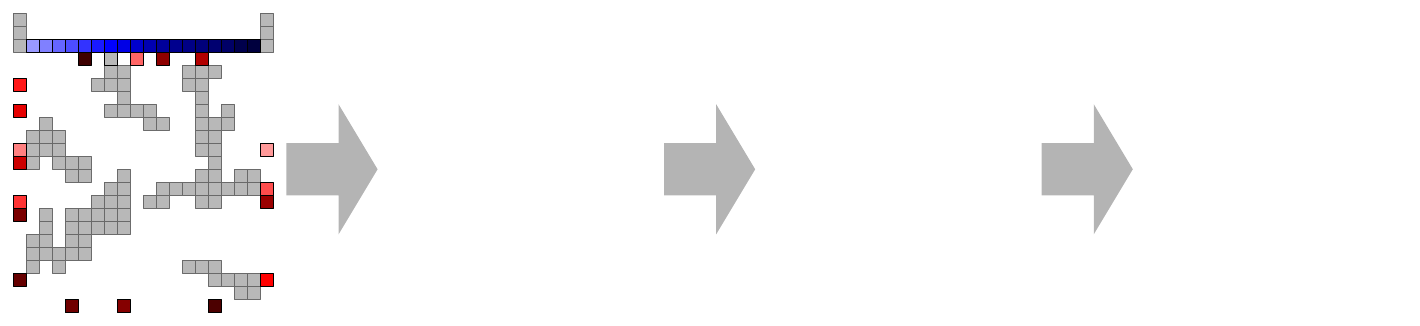}}%
    \put(0.23119257,0.10091744){\makebox(0,0)[t]{\lineheight{1.25}\smash{\begin{tabular}[t]{c}$\mathcal{O}(d)$\end{tabular}}}}%
    \put(0.49724763,0.10091744){\makebox(0,0)[t]{\lineheight{1.25}\smash{\begin{tabular}[t]{c}$\mathcal{O}(d)$\end{tabular}}}}%
    \put(0.76330269,0.10091744){\makebox(0,0)[t]{\lineheight{1.25}\smash{\begin{tabular}[t]{c}$\cdots$\end{tabular}}}}%
    \put(0,0){\includegraphics[width=\unitlength,page=2]{single_boundary_sort.svg.pdf}}%
  \end{picture}%
\endgroup%

	\caption{
		An illustration of the first iteration of the boundary reconfiguration approach:
		We repeatedly apply \textsc{RotateSort} to the top boundary and then rotate by $m-2$ units to gather and sort a sequence of robots (red), ignoring the rest (here: uniform green).
	}
	\label{fig:boundary_circulation}
\end{figure}

\subsection{Exchange between interior and boundary}
\label{subsec:exchange-between-interior-and-boundary}
Having proven that we can freely rearrange both the interior and the boundary individually,~we now prove that we can locally exchange agents between the interior and boundary of a~tile.

\begin{lemma}
	The interiors and boundaries of all tiles in a labeled tiled configuration that consists of $m\times m$-tiles with $m\in \mathcal{O}(d)$ can locally exchange any number of agents by a schedule of makespan $\mathcal{O}(d)$.
	\label{lem:interior_boundary_mixing}
\end{lemma}

\begin{proof}
	We place at most $4m-16$ agents that need to be swapped into the boundary adjacent to the respective boundary agents that need to swap into the interior using~\Cref{lem:labeled_interior_od}.
	Subsequently, we apply~\Cref{the:rotatesort} to the relevant areas.
	As the boundary consists of $4m-4$ agents, at most two repetitions of this process can replace every agent on the boundary.
	This approach may be applied to all tiles simultaneously.
\end{proof}

\section{Schedules of constant stretch}
\label{sec:constantstretch}

In this section, we describe an approach to determine stable schedules with constant stretch for labeled configurations of sufficient \emph{scale}, which we define as follows.

A~configuration $C$ is~\emph{$c$-scaled}, if its silhouette is the union of $c \times c$ squares of robots that are aligned edge-to-edge, i.e., aligned with a $c\times c$ grid.
The \emph{scale} of a configuration $C$ is the maximal $c$ such that $C$ is $c$-scaled.

In the previous section, we introduced tilings and demonstrated the powerful effect of a tiling-based scaffold on local reconfiguration capabilities.
While we know how to exploit these tilings for local movement, we so far lack the tools to build a scaffold and therefore obtain a tiled configuration.
We now investigate a class of configurations based on scale, for which we can efficiently build and make use of a scaffold structure for coordinated motion planning.
In particular, we show the following result.

\smallskip
\begin{theorem}\label{the:c_lower}
    There is a constant $c^*$ such that for any pair $(C_s,C_t)$ of labeled connected~configurations with $n$ robots each and scale of at least $c^*$, there exists a constant stretch schedule $C_s\rightrightarrows_{\chi} C_t$.
\end{theorem}

\medskip
Note that we do not focus on minimizing the constant $c^*$ but rather argue its existence.
An intuitive motivation of this goal is given by~\Cref{the:sqrt-n-stretch}, which demonstrates that for instances of scale $c=1$, constant stretch schedules may not exist.

Our approach works in different phases.
Rather than proving~\Cref{the:c_lower} directly, we prove that each of the phases can be realized by a schedule of constant stretch.

\subsection{Overview of the algorithm}
\label{subsec:overview}

We now give a high-level outline of the purpose of the different phases; an overview is depicted in~\Cref{fig:overview_approx_algorithm}.
Based on preprocessing steps (\Cref{subsec:preprocessing}) in which we determine the scale~$c$ and the diameter~$d$ of the configurations, we build a tile-based scaffolding structure (\Cref{sec:scaffold_construction}) to
support connectivity during the reconfiguration.
The actual reconfiguration consists of shifting robots between adjacent tiles based on flow computations~(\Cref{subsec:interior-flow,,sec:boundary_reordering}), and reconfiguring tiles in parallel~(\Cref{subsec:scaffold-deconstruction}).
Then, deconstructing the scaffold yields the target configuration.

\begin{figure}[th]
    \centering
    \def\svgwidth{\columnwidth}
    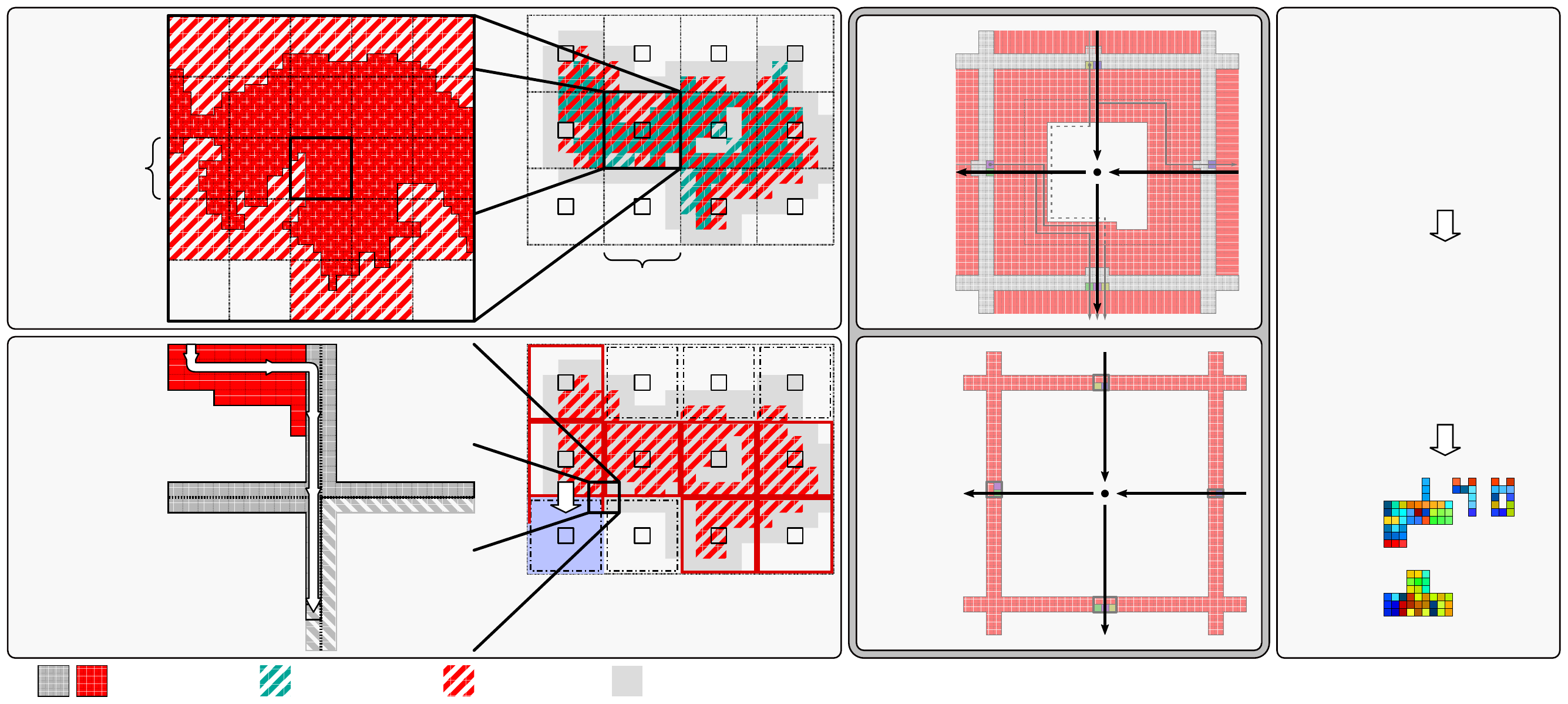
    \caption{Overview of our algorithm, consisting of the phases (1) Preprocessing, (2) Scaffold construction, (3) Interior flow computation and realization, (4) Boundary flow computation and realization, (5) Local tile reconfiguration, followed by (6) Scaffold deconstruction (not shown here).}
    \label{fig:overview_approx_algorithm}
\end{figure}

On a technical level, the phases can be summarized as follows.

\begin{description}
    \item[Phase~\num{1} -- Preprocessing:] Compute scale $c$, diameter $d$,
    a tiling~$\scaffold_1$ of the grid into squares of size $cd\times cd$, and
    a larger tiling~$\scaffold_5$ covering all non-empty tiles of~$\scaffold_1$.
    \item[Phase~\num{2} -- Scaffold construction:] Construct a scaffold along the edges of $\scaffold_5$, resulting in a tiled configuration, guaranteeing connectivity during reconfiguration.
    \item[Phase~\num{3} -- Interior flow:]
    Create a flow graph $G_{\scaffold_5}$ to model the movement of interior (non-scaffold)
    robots between adjacent tiles of~$\scaffold_5$.
    Convert the flow into a series of moves, placing all interior robots in their target tiles.
    \item[Phase~\num{4} -- Boundary flow:]
    Analogously to Phase~\num{3}, create a flow for the robots that were used to construct a scaffold in Phase~\num{2}.
    Subsequently, move all of those robots into the boundary of their target tiles.
    \item[Phase~\num{5} -- Local tile reconfiguration:] Locally reconfigure all tiles of the resulting tiled configuration.
    \item[Phase~\num{6} -- Scaffold deconstruction:] Reverse of Phase~\num{2}.
\end{description}

In the remainder of this section we provide details of the different phases.

\subsection{Phase 1: Preprocessing}
\label{subsec:preprocessing}

Consider the start and target configurations $C_s$ and $C_t$.
Let $c$ refer to the minimum of the two configurations' scales and the as-of-now undetermined constant $c^*$, and let~$d$ be the diameter of the pair $(C_s, C_t)$.
Note that, since we assume $c\geq c^*$, this allows us to treat $c$ as a constant from this point onwards.
We say that two configurations \emph{overlap}, if they have at least one occupied position in common.
For the remainder of the paper, without loss of generality, assume that the configurations overlap in at least one position. 
Otherwise, we simply move the target configuration, such that it overlaps with the start configuration. 
This can be done in $\mathcal{O}(d)$ steps, and results in a new diameter $d' \leq 2d$.

We then define a $cd$-tiling of the underlying grid, which we call $\scaffold_1$.
Let $\scaffold_1(C_s)$ and $\scaffold_1(C_t)$ refer to the set of tiles in $\scaffold_1$ that contain robots in $C_s$ and $C_t$, respectively.
The tiling and respective non-empty tiles can be seen in~\Cref{subfig:cd-tiling,,subfig:neighborhood-relationship}.
\begin{figure}[ht]
    \centering
    \begin{subfigure}[b]{0.3\linewidth}
        \def\svgscale{0.71}
\begingroup%
  \makeatletter%
  \providecommand\color[2][]{%
    \errmessage{(Inkscape) Color is used for the text in Inkscape, but the package 'color.sty' is not loaded}%
    \renewcommand\color[2][]{}%
  }%
  \providecommand\transparent[1]{%
    \errmessage{(Inkscape) Transparency is used (non-zero) for the text in Inkscape, but the package 'transparent.sty' is not loaded}%
    \renewcommand\transparent[1]{}%
  }%
  \providecommand\rotatebox[2]{#2}%
  \newcommand*\fsize{\dimexpr\f@size pt\relax}%
  \newcommand*\lineheight[1]{\fontsize{\fsize}{#1\fsize}\selectfont}%
  \ifx\svgwidth\undefined%
    \setlength{\unitlength}{517.5bp}%
    \ifx\svgscale\undefined%
      \relax%
    \else%
      \setlength{\unitlength}{\unitlength * \real{\svgscale}}%
    \fi%
  \else%
    \setlength{\unitlength}{\svgwidth}%
  \fi%
  \global\let\svgwidth\undefined%
  \global\let\svgscale\undefined%
  \makeatother%
  \begin{picture}(1,0.34057971)%
    \lineheight{1}%
    \setlength\tabcolsep{0pt}%
    \put(0,0){\includegraphics[width=\unitlength,page=1]{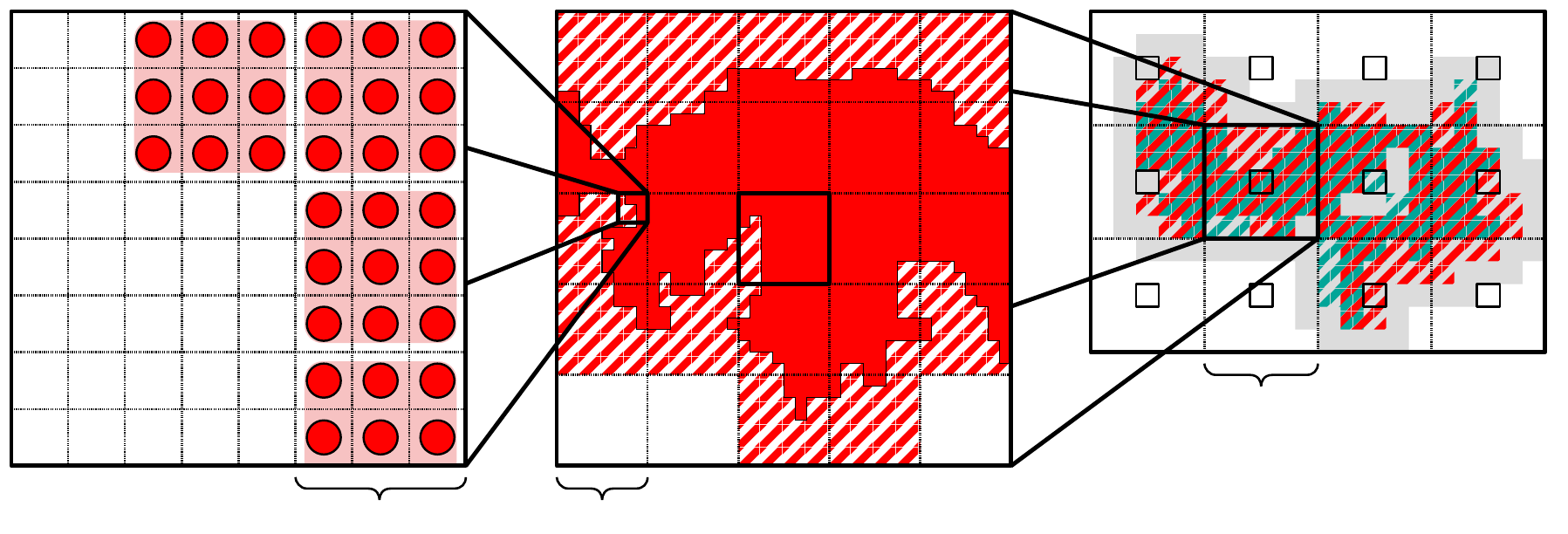}}%
    \put(0.24202899,0.00000001){\makebox(0,0)[t]{\lineheight{1.25}\smash{\begin{tabular}[t]{c}c\end{tabular}}}}%
    \put(0.38405798,0.00000001){\makebox(0,0)[t]{\lineheight{1.25}\smash{\begin{tabular}[t]{c}cd\end{tabular}}}}%
    \put(0.80434783,0.07246376){\makebox(0,0)[t]{\lineheight{1.25}\smash{\begin{tabular}[t]{c}5cd\end{tabular}}}}%
  \end{picture}%
\endgroup%

        \caption{}
    \end{subfigure}\hfil%
    \begin{subfigure}[b]{0.4\linewidth}
        \caption{}
        \label{subfig:cd-tiling}
    \end{subfigure}\hfil%
    \begin{subfigure}[b]{0.3\linewidth}
        \caption{}
        \label{subfig:neighborhood-relationship}
    \end{subfigure}\hfil%
    \caption{
        Using the scale $c$ and diameter of the instance (a), we define the tiling $\scaffold_1$ (b).
        This tiling is further partitioned into $5\times 5$ squares, (c) forming $\scaffold_5$ based on $T_1,\dots,T_\kappa$ (the small squares).
        Those tiles of $\scaffold_5$ which contain the neighbors (in gray) of non-empty start and target tiles (hatched red and green, respectively) of $\scaffold_1$ will now form the basis of our scaffold.
    }
    \label{fig:2-neighborhood_cover}
\end{figure}
Before constructing our scaffold structure, we we introduce a basic tiling, using a finite set of tiles.
As we know that each robot is initially within $d$ units of its target position, we have a particular interest in the tiles of $\scaffold_1(C_s)$ and $\scaffold_1(C_t)$, as well as directly adjacent tiles.
To describe this relationship, let the $k$-\emph{neighborhood} $N_k[\scaffold']$ of a set of tiles refer to all tiles of the same tiling that have~$L_\infty$-distance at most $k$ to any member of $\scaffold'$.
An example of a 1-neighborhood can be seen in~\Cref{subfig:neighborhood-relationship}, where the set $N_1[\scaffold_1(C_s)\cup\scaffold_1(C_t)]$ is marked in gray.
We observe that this is a subset of both $N_2[\scaffold_1(C_s)]$ and
$N_2[\scaffold_1(C_t)]$, because $\scaffold_1(C_s)$ and $\scaffold_1(C_t)$ are fully contained in each other's $1$-neighborhoods.
This relationship is also depicted in~\Cref{subfig:neighborhood-relationship}.

In~\Cref{sec:localreconfiguration}, we discussed previous work that employs a scaffold based on $\mathcal{O}(d)$-tilings to allow for efficient connected reconfiguration of unlabeled robots.
We take a similar approach, effectively drawing robots from either the interior or the $2$-neighborhood of any given tile, in order to construct its boundary.
However, labeled robots have individual target positions that we must take into account when constructing a scaffold.
As a result, we consider an additional, higher-resolution tiling that allows us to ensure that the diameter of the instance does not increase beyond a constant factor  due to scaffold construction.
A detailed description of, and argument for this relationship will be provided in the next section.

Based on the non-empty tiles in $\scaffold_1$, we compute a higher-resolution tiling of the relevant area by a grid of $5\times 5$ squares of $cd$-tiles.
Let $T_1, \dots, T_\kappa$ be a minimal set of tiles from $\scaffold_1$, such that the distance between any two of them is a multiple of $5$ on both the $x$- and $y$-axis, and $N_1[\scaffold_1(C_s)\cup\scaffold_1(C_t)]\subseteq N_2[T_1, \dots, T_\kappa]$.

This directly provides us with the relevant tiles $\scaffold_5 := \{N_2[T_1], \dots, N_2[T_\kappa]\}$, upon which we will base our scaffold and robot motion.
Both $\scaffold_5$ and $T_1, \dots, T_\kappa$ can be seen in~\Cref{subfig:neighborhood-relationship}.
We conclude our preprocessing and move on to the construction of an applicable tiled configuration from $C_s$.

\subsection{Phase 2: Scaffold construction}
\label{sec:scaffold_construction}

Having determined a cover of $C_s$ and $C_t$ in shape of the $5cd$-tiling~$\scaffold_5$, a scaffold spanning this cover is to be constructed.
For this purpose, a robot is to be placed at every boundary position in $\scaffold_5$, which requires $(5\cdot 4cd-4)$ robots per tile.
We show that there are sufficiently many robots to build the scaffold.
We refer to the \emph{locally available material} for a given tile~${T\in\scaffold_5}$ as the set of robots contained within $T$ itself and its immediate neighborhood~$N_1[T]$ over $\scaffold_5$.
As $\scaffold_5$ is a cover of $N_1[\scaffold_1(C_s)\cup\scaffold_1(C_t)]$, we can assume that every tile $T\in\scaffold_5$ contains at least one $T'\in N_2[\scaffold_1(C_s)]$.
We can thus guarantee sufficient locally available material for the boundary of $T$ if we can show that it exists in the $4$-neighborhood of any such~$T'$.

In the context of the scaffold construction, we therefore distinguish \emph{donor} and \emph{recipient}~tiles.
A donor tile contains sufficiently many robots to construct a boundary around itself and all eight immediate neighbors; any other tile is a recipient.
The following lemma shows that every tile of $\scaffold_5$ is either a donor, or an immediate neighbor of an appropriate donor.
These terms will not be used in later sections of the paper, as their concepts are only relevant to this phase of the algorithm.

\begin{lemma}
    \label{lem:sufficient_material}
    There is a constant $c$, such that for all $C_s$ and $C_t$ with scale at least $c$, there is sufficient locally available material to construct a scaffold of~$\scaffold_5$.
\end{lemma}

\begin{proof}
    Unless $cd \in \Omega(\sqrt{n})$, we can assume that both the start and target configurations consist of more than a single tile in $\scaffold_1$;
    otherwise, the reconfiguration problem is trivial.
    Without loss of generality, we assume that for any given non-empty tile~$T\in\scaffold_1$, there exist robots outside of its $1$-neighborhood in $C_s$.
    This implies the existence of a path of length at least $cd$ that connects a position of $T$ to those robots.
    Such a path has to be contained in a union of fully occupied pixels of size $c\times c$ due to the scale of both configurations.
    We can employ a proof by Fekete et al.~\cite{connected-motion-journal} to derive that at least one neighbor of $T$ must contain no less than $\nicefrac{(c^2 d)}{4}$ robots for $c\geq4$.

    Because $\scaffold_5$ is a cover of the $2$-neighborhood of non-empty tiles over $\scaffold_1$, we extend the above finding to conclude that every tile in~$\scaffold_5$, or one of its neighbors, contains at least $\nicefrac{(c^2 d)}{4}$ robots.
    We additionally note that the respective robots are located within the $4$-neighborhood of each recipient tile over $\scaffold_1$.
    This is of particular relevance, as these robots will remain within the $1$-neighborhood over $\scaffold_5$ for any target configuration.

    A worst-case scenario has a tile donating to all eight of its neighbors, so it must contain at least $9(5\cdot 4cd-4)$ robots.
	We can guarantee sufficient material based on the above relationship, starting from a lower bound of $\nicefrac{(c^2 d)}{4} \geq 9(20cd-4) \Leftrightarrow c \geq 720$.
As this proof is the only portion of our algorithm that hinges on scale, the constant~$c^*$ from~\Cref{the:c_lower} is no greater than $720$. 
\end{proof}

The scaffold construction now takes place in the following three subphases.

\begin{description}
    \item[Phase 2.1:] Parallel construction of all donor boundaries.
    \item[Phase 2.2:] Mapping recipients to donors appropriately.
    \item[Phase 2.3:] Construction of recipient boundaries.
\end{description}

\begin{lemma}
	\label{lem:scaffold_od}
    Constructing the scaffold (i.e., Phase 2 of our approach) takes no more than~$\mathcal{O}(d)$ transformations.
\end{lemma}

\begin{proof}
We show this by providing an overview of each subphase, and arguing that no subphase takes more than $\mathcal{O}(d)$ transformations.

For \emph{Phase 2.1}, consider a donor tile $T\in\scaffold_5$.
As each donor's interior contains sufficiently many robots to build its own boundary, we can employ a very simple strategy to construct their boundaries.
Within any such tile, we assign every robot a priority based on the length of the shortest path in $C_s$ that connects the robot to another robot that is adjacent to an empty position on the boundary.
Until $T$'s boundary is fully occupied, we repeatedly push along the respective shortest path of a robot with maximal priority in $T$, adding one robot to the scaffold structure, as depicted in~\Cref{fig:scaffold_construction-a}.

\begin{figure}[htb]
	\begin{subfigure}[b]{0.67\linewidth}
		\centering
		\def\svgscale{0.7}
\begingroup%
  \makeatletter%
  \providecommand\color[2][]{%
    \errmessage{(Inkscape) Color is used for the text in Inkscape, but the package 'color.sty' is not loaded}%
    \renewcommand\color[2][]{}%
  }%
  \providecommand\transparent[1]{%
    \errmessage{(Inkscape) Transparency is used (non-zero) for the text in Inkscape, but the package 'transparent.sty' is not loaded}%
    \renewcommand\transparent[1]{}%
  }%
  \providecommand\rotatebox[2]{#2}%
  \newcommand*\fsize{\dimexpr\f@size pt\relax}%
  \newcommand*\lineheight[1]{\fontsize{\fsize}{#1\fsize}\selectfont}%
  \ifx\svgwidth\undefined%
    \setlength{\unitlength}{333.75bp}%
    \ifx\svgscale\undefined%
      \relax%
    \else%
      \setlength{\unitlength}{\unitlength * \real{\svgscale}}%
    \fi%
  \else%
    \setlength{\unitlength}{\svgwidth}%
  \fi%
  \global\let\svgwidth\undefined%
  \global\let\svgscale\undefined%
  \makeatother%
  \begin{picture}(1,0.47191011)%
    \lineheight{1}%
    \setlength\tabcolsep{0pt}%
    \put(0,0){\includegraphics[width=\unitlength,page=1]{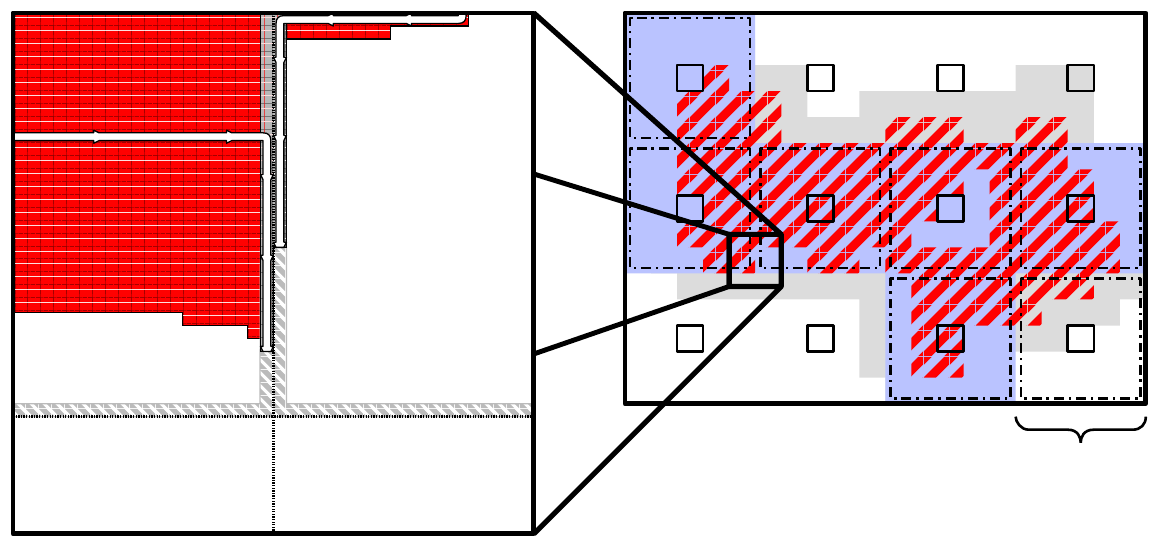}}%
    \put(0.93258422,0.05617983){\makebox(0,0)[t]{\lineheight{1.25}\smash{\begin{tabular}[t]{c}5cd\end{tabular}}}}%
    \put(0,0){\includegraphics[width=\unitlength,page=2]{scaffold_construction_a.svg.pdf}}%
  \end{picture}%
\endgroup%

		\caption{}
		\label{fig:scaffold_construction-a}
	\end{subfigure}
	\begin{subfigure}[b]{0.32\linewidth}
		\centering
		\def\svgscale{0.7}
\begingroup%
  \makeatletter%
  \providecommand\color[2][]{%
    \errmessage{(Inkscape) Color is used for the text in Inkscape, but the package 'color.sty' is not loaded}%
    \renewcommand\color[2][]{}%
  }%
  \providecommand\transparent[1]{%
    \errmessage{(Inkscape) Transparency is used (non-zero) for the text in Inkscape, but the package 'transparent.sty' is not loaded}%
    \renewcommand\transparent[1]{}%
  }%
  \providecommand\rotatebox[2]{#2}%
  \newcommand*\fsize{\dimexpr\f@size pt\relax}%
  \newcommand*\lineheight[1]{\fontsize{\fsize}{#1\fsize}\selectfont}%
  \ifx\svgwidth\undefined%
    \setlength{\unitlength}{157.5bp}%
    \ifx\svgscale\undefined%
      \relax%
    \else%
      \setlength{\unitlength}{\unitlength * \real{\svgscale}}%
    \fi%
  \else%
    \setlength{\unitlength}{\svgwidth}%
  \fi%
  \global\let\svgwidth\undefined%
  \global\let\svgscale\undefined%
  \makeatother%
  \begin{picture}(1,1)%
    \lineheight{1}%
    \setlength\tabcolsep{0pt}%
    \put(0,0){\includegraphics[width=\unitlength,page=1]{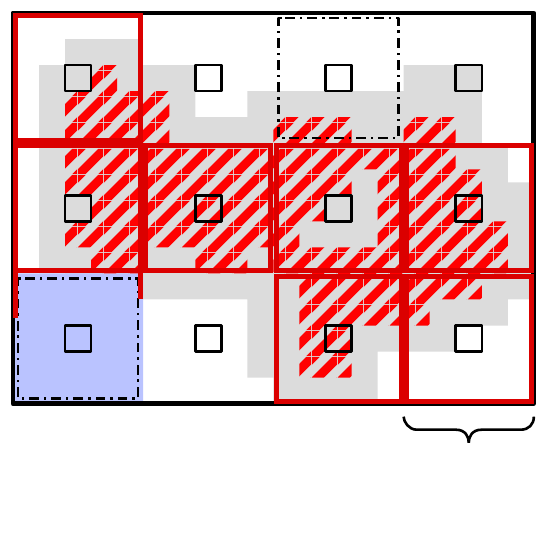}}%
    \put(0.85714292,0.1190476){\makebox(0,0)[t]{\lineheight{1.25}\smash{\begin{tabular}[t]{c}5cd\end{tabular}}}}%
    \put(0,0){\includegraphics[width=\unitlength,page=2]{scaffold_construction_b.svg.pdf}}%
  \end{picture}%
\endgroup%

		\caption{}
		\label{fig:scaffold_construction-b}
	\end{subfigure}
	\caption{
		The first and final subphases of the scaffold construction visualized.
		Tiles with ongoing scaffold construction are highlighted in blue, while dashed boundaries indicate the state of the scaffold after each phase's termination.
	}
	\label{fig:scaffold_construction}
\end{figure}

This approach does not cause disconnections in the configuration, as it already has been shown for the unlabeled problem variant.
Because the scaffold consists of $20cd-4$ robots per tile, Phase~2.1 completes after $\mathcal{O}(d)$ steps.

By~\Cref{lem:sufficient_material}, we already know that an applicable donor exists for each recipient tile.
Therefore, \emph{Phase~2.2} simply focuses on determining an exact mapping between donors and recipients, while ensuring that the correct robots are distributed in the final subphase.
We achieve this by solving a flow problem in an extended version of the dual graph of~$\scaffold_5$, i.e., a graph in which each vertex corresponds to a tile, and two vertices share an edge if the corresponding tiles are horizontally, vertically, or diagonally adjacent, as indicated in~\Cref{fig:scaffold_assignment-a}.
We model the necessary assignment of robots from donors to recipients as a flow network with source vertex $v_{s}$ and sink vertex $v_{t}$.
A maximal flow in this network can be decomposed into a set of (augmenting) paths that connect one robot from a donor tile to one adjacent recipient tile.
This means that each path has the form $(v_s, r, R, v_t)$, with~$r$ being a robot currently located in a donor tile $D\in\scaffold_5$ and a recipient tile $R\in\scaffold_5$.
All edges of type $(v_s, r)$ or $(r, R)$ have capacity one.
Conversely, every edge~$(R, v_t)$ has capacity $20cd-4$, see~\Cref{fig:scaffold_assignment}.

\begin{figure}[h]
	\begin{subfigure}[b]{162bp}
		\def\svgscale{0.9}
		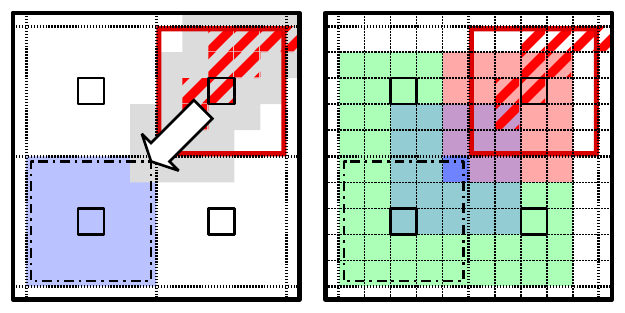
		\caption{}
		\label{fig:scaffold_4-neighborhood}
	\end{subfigure}\hfil%
	\begin{subfigure}[b]{74.25bp}
		\def\svgscale{0.9}
		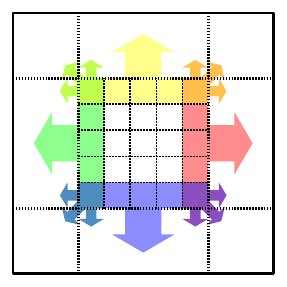
		\caption{}
		\label{fig:scaffold_assignment-a}
	\end{subfigure}\hfil%
	\begin{subfigure}[b]{114.75bp}
		\def\svgscale{0.9}
\begingroup%
  \makeatletter%
  \providecommand\color[2][]{%
    \errmessage{(Inkscape) Color is used for the text in Inkscape, but the package 'color.sty' is not loaded}%
    \renewcommand\color[2][]{}%
  }%
  \providecommand\transparent[1]{%
    \errmessage{(Inkscape) Transparency is used (non-zero) for the text in Inkscape, but the package 'transparent.sty' is not loaded}%
    \renewcommand\transparent[1]{}%
  }%
  \providecommand\rotatebox[2]{#2}%
  \newcommand*\fsize{\dimexpr\f@size pt\relax}%
  \newcommand*\lineheight[1]{\fontsize{\fsize}{#1\fsize}\selectfont}%
  \ifx\svgwidth\undefined%
    \setlength{\unitlength}{127.5bp}%
    \ifx\svgscale\undefined%
      \relax%
    \else%
      \setlength{\unitlength}{\unitlength * \real{\svgscale}}%
    \fi%
  \else%
    \setlength{\unitlength}{\svgwidth}%
  \fi%
  \global\let\svgwidth\undefined%
  \global\let\svgscale\undefined%
  \makeatother%
  \begin{picture}(1,0.64705882)%
    \lineheight{1}%
    \setlength\tabcolsep{0pt}%
    \put(0,0){\includegraphics[width=\unitlength,page=1]{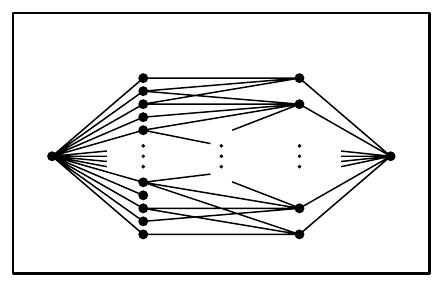}}%
    \put(0.11764707,0.35294117){\color[rgb]{0,0,0}\makebox(0,0)[t]{\lineheight{1.25}\smash{\begin{tabular}[t]{c}$v_s$\end{tabular}}}}%
    \put(0.88235303,0.35294117){\color[rgb]{0,0,0}\makebox(0,0)[t]{\lineheight{1.25}\smash{\begin{tabular}[t]{c}$v_t$\end{tabular}}}}%
    \put(0.32941176,0.52735715){\color[rgb]{0,0,0}\makebox(0,0)[t]{\lineheight{1.25}\smash{\begin{tabular}[t]{c}robots\end{tabular}}}}%
    \put(0.68243913,0.52648783){\color[rgb]{0,0,0}\makebox(0,0)[t]{\lineheight{1.25}\smash{\begin{tabular}[t]{c}$\scaffold_5$\end{tabular}}}}%
  \end{picture}%
\endgroup%

		\caption{}
		\label{fig:scaffold_assignment-b}
	\end{subfigure}
    \caption{
		For the scaffold construction, we may have to move robots to some $T\in\scaffold_1$ that is empty, but in $N_2[\scaffold_1(C_s)]$, see the marked grey tile in (a).
		Sufficient robots exist in $N_2[T']$ (red) for~${T'\in N_2[T]}$, i.e., in $N_4[T]$ (green).
		As robots only need to travel between adjacent tiles of $\scaffold_1$ (b), we can select them by solving a flow problem (c).
    }
    \label{fig:scaffold_assignment}
\end{figure}

Because there is a sufficient number of robots in each recipient tile's vicinity, a maximal flow fully utilizes all edges that connect recipients to the sink vertex $v_t$.
The selection of edges between robots and recipients thus determines an applicable~assignment.
Note that, due to~\Cref{lem:scaffold_od}, we use robots to build a scaffold for a recipient tile from its $4$-neighborhood over $\scaffold_1$ of an applicable donor.
As shown in~\Cref{fig:scaffold_4-neighborhood,fig:scaffold_assignment-a}, this implies that robots from the donor's ``far side'' (consisting of the $\scaffold_1$ tiles where robots may seek to leave into the opposite direction) are not used; otherwise, they would move too far away from their~target~location.

As the construction scheme relies solely on the silhouette of a donating tile, we can easily determine to which adjacent recipient a robot at any given position in a donor tile will be sent.
This includes scaffold locations near the boundary between two tiles that may have been 
occupied by robots with target neighbors opposite to the recipient tile.
We make use of the concepts of~\Cref{the:local_od} to gather specific robots from a donor for each recipient, again taking $\mathcal{O}(d)$ transformations.
Note that, while~\Cref{the:local_od} may not be applicable directly, since~\Cref{lem:parallel_boundary_reordering} requires a valid tiled configuration, we can employ robots from the donors' interior as substitute filling for the necessary sorting area.
In practice, this corresponds to building a second, immediately adjacent boundary inside each donor tile.

\emph{Phase~2.3} can then safely proceed to push robots from a donor tile into the boundary of adjacent recipients, using a variant of the method from Phase~2.1.
We consider $3\cdot 3$ classes of tiles in $\scaffold_5$, based on each tile's anchor position modulo $15cd$ and construct the boundaries of all recipient tiles in a class in parallel, see~\Cref{fig:scaffold_construction-b}. 
For the movement of each robot, we map robots from the interior of a donor tile to a vertex on the respective recipient's boundary.
Based on the previously discussed priority, robots are now pushed along their boundary and onto boundary positions of the adjacent recipient.
Once this process terminates, we have constructed a scaffold structure around all tiles of the cover~$\scaffold_5$.
This takes $20cd-4$ transformations per tile class, so we conclude that~\Cref{lem:scaffold_od} holds.
\end{proof}

With the help of this global support structure, connectivity is ensured during  the actual reconfiguration.
It remains to show how we shift robots between tiles, and reconfigure robot arrangements within the constructed scaffold.
The latter has already been described in~\Cref{sec:localreconfiguration}, so we only need to describe how to relocate robots between tiles.

\subsection{Phase~\num{3}: Interior flow}
\label{subsec:interior-flow}
This is modeled as a supply-demand flow for interior robots in three subphases:

\begin{description}
	\item[Phase 3.1:] Interior flow computation.
	\item[Phase 3.2:] Interior flow partition.
	\item[Phase 3.3:] Interior flow realization.
\end{description}

Note that we focus on the transfer of interior robots in this section.
A similar approach is used to model the flow for boundary robots, see~\Cref{sec:boundary_reordering}.
We proceed by introducing some fundamental terminology relating to the supply-demand flow in question.

We know that, after the scaffold construction completes, any robot is located within the~$1$-neighborhood of the tile containing its target position.
The desired exchange of robots between tiles is represented as a supply-demand flow $G_{\scaffold_5}\coloneqq~(\scaffold_5, E_{\scaffold_5}, \mathit{f}_{\scaffold_5})$ over~$\scaffold_5$, in which every tile is adjacent to those in its 1-neighborhood.
The flow value of an edge $\mathit{f}_{\scaffold_5}(e)$ corresponds to the cardinality of the set of robots that need to move from one tile into another.
For simplicity, we consider the number of robots in the flow model, rather than the specific robots themselves.
Clearly, the flow on any given edge is bounded from above by the interior space of the tiles, i.e., $(5cd-2)^2$ robots.

A tile $T\in \scaffold_5$ is a \emph{source} (\emph{sink}) if and only if the sum of flow values of incoming edges is smaller (larger) than the sum of flow values of outgoing edges.
Otherwise, we call $T$ \emph{flow-conserving}.
The difference of these sums is called the \emph{supply} and \emph{demand} at sources and sinks, respectively.

If the flow value of every edge within a given flow graph is bounded from above by some value $\sigma$, we refer to it as a \emph{$\sigma$-flow}, e.g., $G_{\scaffold_5}$ is a $\mathcal{O}(d^2)$-flow.
Additionally, we define an~$(a,b)$\emph{-partition} of a flow graph as a set that contains $b$ many $a$-flows that sum up to the original flow.

We say that a schedule \emph{realizes} a flow graph $G_{\scaffold_5} = (\scaffold_5, E_{\scaffold_5}, \mathit{f}_{\scaffold_5})$ if the number of robots moved along any directed edge $e\in E_{\scaffold_5}$ is exactly $\mathit{f}_{\scaffold_5}(e)$.
By construction, the realization of~$G_{\scaffold_5}$ never requires us to fill a tile over capacity, or to remove robots from an empty tile, as this would imply an invalid start or target configuration.
\smallskip
\begin{lemma}
	It is possible to efficiently compute a stable schedule of makespan $\mathcal{O}(d)$ that realizes $G_{\scaffold_5}$.
	\label{lem:flow_realization}
\end{lemma}
\smallskip
Having computed the flow graph $G_{\scaffold_5}$ as above, we perform additional preprocessing steps before handling the remaining two subphases, as outlined in the following sections, culminating in a proof of~\Cref{lem:flow_realization}.

\subsubsection{Phase~\num{3.1}: Interior flow preprocessing}
\label{subsec:flow_preprocessing}
Our partition methods require a planar, unidirectional flow network.
We can obtain both properties by eliminating bidirectional and crossing edges in the flow graph~$G_{\scaffold_5}$, see~\Cref{fig:flow_preprocessing} for the underlying idea.
Note that this process may triple the maximum flow value over every edge, resulting in a $3(5cd-2)^2$-flow.

\begin{figure}[h]
	\begin{subfigure}[c]{148bp}
		\centering
		\def\svgscale{0.9}
		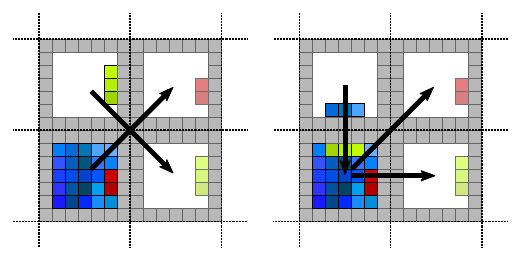
		\caption{}
		\label{fig:flow_preprocessing-a}
	\end{subfigure}\hfil%
	\begin{subfigure}[c]{101.25bp}
		\centering
		\def\svgscale{0.9}
		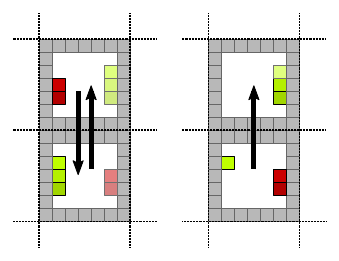
		\caption{}
		\label{fig:flow_preprocessing-b}
	\end{subfigure}\hfil%
	\begin{subfigure}[c]{111.375bp}
		\centering
		\def\svgscale{0.9}
		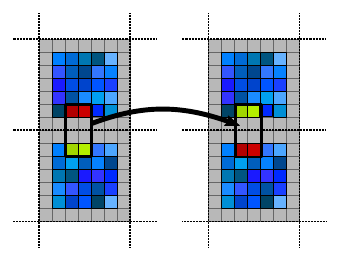
		\caption{}
		\label{fig:local-swap}
	\end{subfigure}
	\caption{Preprocessing steps remove both crossing and bidirectional edge (a+b), by applying \textsc{RotateSort} locally to swap robots through the boundary (c).}
	\label{fig:flow_preprocessing}
\end{figure}

\paragraph*{Removing crossing diagonal edges}
Crossing edges can only occur between two diagonal edges of adjacent source tiles $u,v\in G_{\scaffold_5}$.
To remove such a crossing, it suffices to eliminate one of the two edges by exchanging robots between $u$ and $v$.
We refer to~\Cref{fig:flow_preprocessing-a} for an illustration.
Let $w$ be a common neighbor of $u$ and $v$, and consider without loss of generality that there is a diagonal edge $(u,w)$ that has to be eliminated.
To this end, we simply exchange robots between the source tiles $u,v$ in such a way that the flow $ \mathit{f}_{\scaffold_5}((u,w))$ is rerouted through $v$.
For this, we exchange the robots in the interior of $u$ that want to end up in $w$ with robots of $v$ that either want to stay in $v$ or has their target location in $w$, too.

This increases the flow on the edges $(u,v)$ and $(v,w)$ by at most~$(5cd-2)^2$.
Because there are only two possible tiles that are diagonal adjacent to $v$ and adjacent to $w$, this can occur only two times.
Note that by removing diagonal edges, no new diagonal edge can~occur.
After the removal of all crossing edges, $G_{\scaffold_5}$ is a $3(5cd-2)^2$-flow.

\paragraph*{Removing bidirectional edges}
A bidirectional edge between two adjacent tiles $u,v\in G_{\scaffold_5}$ can be removed by exchanging robots between $u$ and $v$ that want to switch to the other tile until one of the edges disappears.
See~\Cref{fig:flow_preprocessing-b} for an illustration.

\subsubsection{Phase~\num{3.2}: Interior flow partition}
\label{subsec:flow_partitioning}

Having obtained a planar flow graph, we observe the existence of two types of movement;
tiles may exchange robots in a circular motion to move ``misplaced'' robots into their target tiles, without changing the number of robots in the tiles involved,
or the number of robots in a tile may need to change, as the target configuration requires increased density in its vicinity.

We employ two existing results for flow partitioning that deal with the two cases disjointly.
To this end, we split $G_{\scaffold_5}$ into two distinct subproblems, by partitioning it into two flows $G_\rightarrow$ and~$G_\bigcirc$, one being acyclic and the other totally cyclic, see~\Cref{fig:flow_decomposition}.

\begin{figure}[h]
	\centering
	\def\svgwidth{0.75\columnwidth}
\begingroup%
  \makeatletter%
  \providecommand\color[2][]{%
    \errmessage{(Inkscape) Color is used for the text in Inkscape, but the package 'color.sty' is not loaded}%
    \renewcommand\color[2][]{}%
  }%
  \providecommand\transparent[1]{%
    \errmessage{(Inkscape) Transparency is used (non-zero) for the text in Inkscape, but the package 'transparent.sty' is not loaded}%
    \renewcommand\transparent[1]{}%
  }%
  \providecommand\rotatebox[2]{#2}%
  \newcommand*\fsize{\dimexpr\f@size pt\relax}%
  \newcommand*\lineheight[1]{\fontsize{\fsize}{#1\fsize}\selectfont}%
  \ifx\svgwidth\undefined%
    \setlength{\unitlength}{333.75bp}%
    \ifx\svgscale\undefined%
      \relax%
    \else%
      \setlength{\unitlength}{\unitlength * \real{\svgscale}}%
    \fi%
  \else%
    \setlength{\unitlength}{\svgwidth}%
  \fi%
  \global\let\svgwidth\undefined%
  \global\let\svgscale\undefined%
  \makeatother%
  \begin{picture}(1,0.24719101)%
    \lineheight{1}%
    \setlength\tabcolsep{0pt}%
    \put(2.71786251,0.34761659){\color[rgb]{0,0,0}\makebox(0,0)[lt]{\begin{minipage}{2.12333009\unitlength}\end{minipage}}}%
    \put(0,0){\includegraphics[width=\unitlength,page=1]{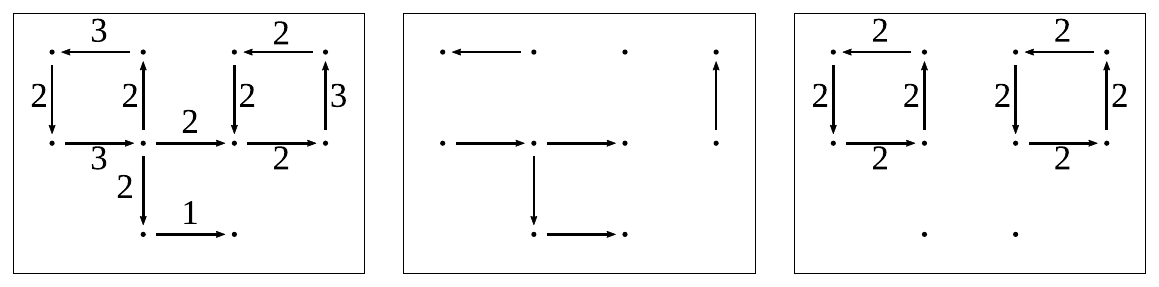}}%
    \put(0.33258427,0.11460675){\color[rgb]{0,0,0}\makebox(0,0)[t]{\lineheight{1.25}\smash{\begin{tabular}[t]{c}=\end{tabular}}}}%
    \put(0.66966292,0.11460675){\color[rgb]{0,0,0}\makebox(0,0)[t]{\lineheight{1.25}\smash{\begin{tabular}[t]{c}+\end{tabular}}}}%
    \put(0,0){\includegraphics[width=\unitlength,page=2]{flow_decomposition.svg.pdf}}%
  \end{picture}%
\endgroup%

	\caption{A decomposition of a flow graph into its acyclic and cyclic components.}
	\label{fig:flow_decomposition}
\end{figure}

\begin{lemma}
	\label{lem:flow_component_separation}
	It is possible to efficiently compute a partition of $G_{\scaffold_5}$ into an acyclic component $G_\rightarrow$ and a totally cyclic component $G_\bigcirc$.
\end{lemma}
\smallskip

This is a standard result from the theory of network flows; e.g., see the books by Korte and Vygen~\cite{korte2011combinatorial} or by Ford and Fulkerson~\cite{ford1962flows}.

\paragraph*{Partitioning the cyclic flow}
For the case of configurations that do not have to be connected, Demaine et
al.~\cite{dfk+-cmprs-19} considered tiles of side length $24d$. 
Thus, they obtained a totally cyclic flow graph $G_\bigcirc$ with an upper bound of $24 d \cdot 24d = 576d^2$ for the flow value of each edge. 
Furthermore, they showed that it is possible to compute a $(d,\mathcal{O}(d))$-partition of~$G_\bigcirc$.
In our case, we have to keep configurations connected, resulting in tiles of side length $c d$.
We extend their peeling algorithm to a broader, parameterized version.

\smallskip
\begin{lemma}\label{lem:circulation_partitioning}
	A $(d,\mathcal{O}(d))$-partition of the totally cyclic $(k\cdot d^2)$-flow~$G_\bigcirc$ for~$k\in\mathbb{N}$ into totally cyclic flows can be computed efficiently.
\end{lemma}

\begin{proof}
	Consider a tiled configuration $C'$ of scale $c$ and let $n$ denote the total number of robots.
	The following method yield a $(d,\mathcal{O}(d))$-partition of $G_\bigcirc$.
	
	We start with computing a $(1,h)$-partition $\mathbb{P}_\bigcirc$ of $G_\bigcirc$.
	The number of flows $h$ is bounded from above by the total number of robots $n$, as each robot may only
	contribute to a single cycle.
	If any cycle intersects itself, it is subdivided into smaller, non-intersecting cycles.
	Subsequently, $\mathbb{P}_\bigcirc$ can be divided into two classes of cycles based on cycle orientation.
	Let~$\mathbb{P}_\circlearrowright$ refer to the clockwise cycles, and $\mathbb{P}_\circlearrowleft$ to the
	counterclockwise cycles of $\mathbb{P}_\bigcirc$.
	
	Subsequently, we compute a $(1, h)$-partition using the sets $\mathbb{P}_\circlearrowright$ and $\mathbb{P}_\circlearrowleft$.
	The resulting subsets $\mathbb{P}_\circlearrowright^1 \cup \mathbb{P}_\circlearrowright^2 \cup \mathbb{P}_\circlearrowleft^1 \cup \mathbb{P}_\circlearrowleft^2$ are constructed in a fashion such that two cycles $u$ and $v$ from the same subset share the same orientation and are either edge-disjoint or nested within another, where we write $u \sqsubseteq v$ if $v$ contains $u$.
	
	Such a partition may be constructed using the geometric peeling algorithm, depicted in~\Cref{fig:peeling_approach}.
	Their approach handles each $\mathbb{P}_X\in\{\mathbb{P}_\circlearrowright, \mathbb{P}_\circlearrowleft\}$
	separately.
	Considering the union of inner-bound areas $A$ of all cycles in $\mathbb{P}_X$, a 1-circulation around the outer edges of each component of $A$ is identified, removed from $\mathbb{P}_X$, and added to the set~$\mathbb{P}_X^1$.
	Simultaneously, inversely oriented 1-circulations around any holes within~$A$ are removed from $\mathbb{P}_X$, and added to the set $\mathbb{P}_X^2$, see~\Cref{fig:geometric_peeling}.
	This process is repeated until $\mathbb{P}_X$ is fully decomposed into the two sets $\mathbb{P}_X^1$ and $\mathbb{P}_X^2$.
	
	\begin{figure}[htb]
		\centering
		\begin{subfigure}[c]{\linewidth}
			\centering
			\def\svgwidth{0.9\columnwidth}
			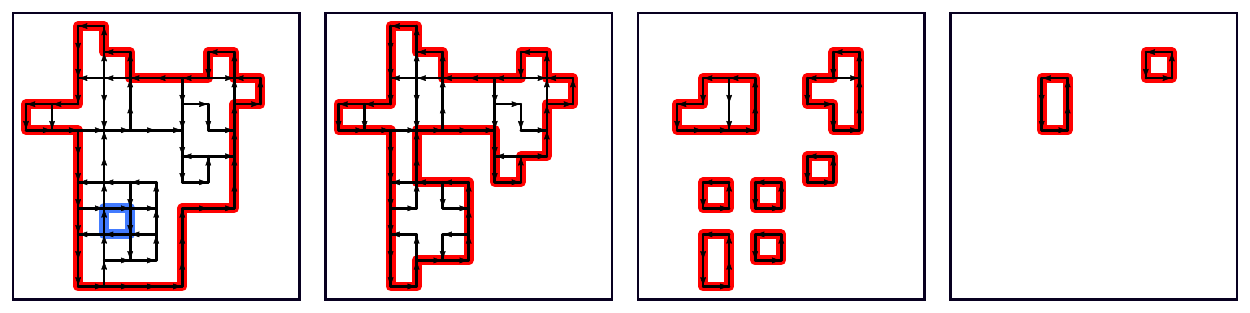
			\caption{The area that is considered during each iteration of the peeling algorithm.}
			\label{fig:geometric_peeling}
		\end{subfigure}
		\begin{subfigure}[c]{\linewidth}
			\centering
			\def\svgwidth{0.9\columnwidth}
			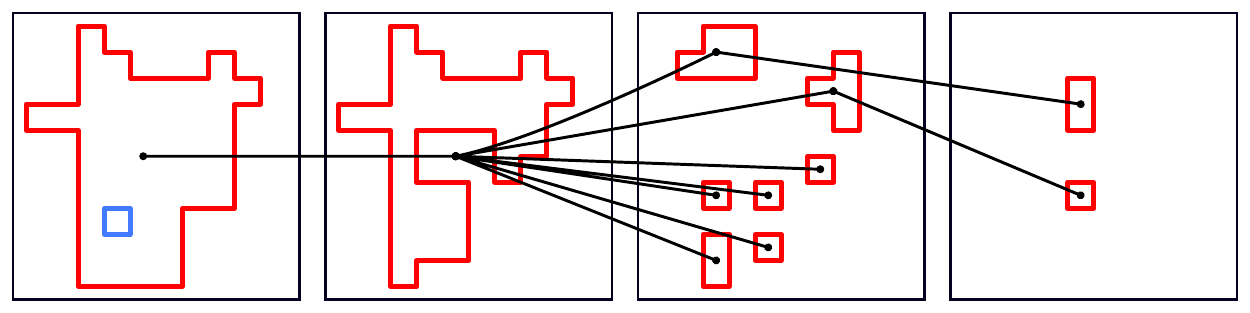
			\caption{The tree induced by the nesting properties, grouped by depth.}
			\label{fig:geometric_peeling_tree}
		\end{subfigure}
		\caption[Geometric Peeling Algorithm]{Iterations of the peeling algorithm. The outer loops that go
			into a set $\mathbb{P}_X^1$ are marked red and the inner boundary loops of set $\mathbb{P}_X^2$ are marked in
			blue.}
		\label{fig:peeling_approach}
	\end{figure}
	
	Finally, each $\mathbb{P}_X^j\in \{\mathbb{P}_\circlearrowright^1, \mathbb{P}_\circlearrowright^2, \mathbb{P}_\circlearrowleft^1, \mathbb{P}_\circlearrowleft^2\}$ is now partitioned into $\mathcal{O}(d)$ sets that form one $d$-subflow of $G_\bigcirc$ each.
	
	As previously noted, any pair of cycles in $\mathbb{P}_X^j$ is either edge-disjoint, or one of the cycles lies within the other.
	This property induces a forest $F = (\mathbb{P}_X^j, E_F)$, as depicted in~\Cref{fig:geometric_peeling_tree}, in which one cycle $u$ is a child of another cycle $v$ exactly if $u\sqsubseteq v$ and there is no other cycle~$w$ such that $u\sqsubseteq w\sqsubseteq v$.
	
	Given such a forest, every vertex is labeled by its depth in $F \bmod kd$.
	These labels are then used to construct $\mathcal{O}(d)$ subflows $\{G_{\bigcirc_1},G_{\bigcirc_2}, \dots \}$, each
	$G_{\bigcirc_i}$ being the union of cycles with label $i$.
	
	It only remains to be shown that every flow $G_{\bigcirc_i}$, constructed by this algorithm, is a $d$-subflow of $G_\bigcirc$.
	For this, consider any $u,v\in \mathbb{P}_X^j$ such that the two cycles share a common edge $e$.
	By construction, one of the two cycles must lie within the other, so without loss of generality, assume that $u\sqsubseteq v$.
	This implies that there exists a path from $u$ to its root via $v$ in $F$, with all cycles on the path between $u$ and $v$ sharing the edge $e$ as well.
	As $G_\bigcirc$ has all edges bounded from above by $k\cdot d^2$, the cycles containing~$e$ lie on a path of length no more than $k\cdot d^2$ in $F$.
	Thus, $e$ has a weight of at most $\nicefrac{(k\cdot d^2)}{(k\cdot d)} = d$ in every $G_{\bigcirc_i}$, meaning $G_{\bigcirc_i}$ is a $d$-flow.
\end{proof}

Applying this algorithm to $G_\bigcirc$ yields a $(d, \ell)$-partition for $\ell < 75c^2\cdot d = \mathcal{O}(d)$.
As~the elements of such a partition of $G_\bigcirc$ add up to $G_\bigcirc$ itself, its realization may now be reduced to the realization of all $d$-subflows.

\paragraph*{Partitioning the acyclic flow}

We already computed a $(d,\mathcal{O}(d))$-partition of $G_\bigcirc$. 
As $G_{\scaffold_5} = G_\bigcirc + G_\rightarrow$, we still need to compute a $(d,\mathcal{O}(d))$-partition of $G_\rightarrow$ to obtain a $(d,\mathcal{O}(d))$-partition of the entire flow $G_{\scaffold_5}$.
In the context of  unlabeled robots, Fekete et al.~\cite{connected-motion-journal} proposed an algorithm for computing a $(\mathcal{O}(d^2), 28)$-partition of $G_\rightarrow$.
Due to the much more complex situation of labeled robots, the underlying ideas need to be refined to provide an algorithm that guarantees the following.

\smallskip
\begin{lemma}\label{lem:supplydemand_partitioning}
	A $(d,\mathcal{O}(d))$-partition of the acyclic $(k\cdot d^2)$-flow~$G_\rightarrow$ for $k\in\mathbb{N}$ into acyclic flows can be computed efficiently.
\end{lemma}

\begin{proof}
	It is possible to compute a partition of $G_\rightarrow = ( \scaffold_5, E, \supplydemandflow)$ into $1$-subflows, each being a path that connects a supply vertex to a demand vertex, simply by performing a decomposition of the flow.
	As $G_\rightarrow$ is both planar and unidirectional, this means it may be viewed as a directed forest, with each component $A$ containing a subset of the paths.
	The following process is then applied to each tree $A\subseteq G_\rightarrow$ separately.
	
	First, choose an arbitrary vertex of $A$ as the root of the component.
	For any path~$P_i$ in $A$, its \emph{link distance} refers to the minimal length of any path connecting a vertex of~$P_i$ to the root of $A$.
	Consider $(P_1, P_2, \dots)$ as sorted by link distance.
	Every path~$P_i$ is now greedily assigned to a set $S_j$, such that the first edge $e_1$ of $P_i$ is not part of any other path in $S_j$.
	These sets are shared over the components of $G_\rightarrow$, and if no matching set exists for a path~$P_i$, a new set $S_j := \{P_i\}$ is created.
	As the maximum incoming degree of the head vertex of a directed edge is three in the setting of a grid graph, each edge of the graph may be contained in at most three paths of each set.
	
	Once these sets are fully constructed, they are greedily partitioned into groups of $\nicefrac{d}{3}$ sets $\{G_{\rightarrow_1}, G_{\rightarrow_2}, \dots\}$.
	For each group $G_{\rightarrow_i}$, a subflow $\mathit{f}_{i}$ is created that maps an edge to the number of paths in $G_{\rightarrow_i}$ that contain it.
	Because every edge can appear no more than three times within a set of paths $S_j$ and each group contains $\nicefrac{d}{3}$ sets, each $f_i$ is a $(\nicefrac{d}{3}\cdot 3) = d$-subflow of $G_\rightarrow$.
	
	It only remains to show that the constructed partition $( G_{\rightarrow_1}, G_{\rightarrow_2}, \dots )$ contains at most~${75c^2 d = \mathcal{O}(d)}$ subflows.
	As $G_{\scaffold_5}$ is initially bounded from above by the interior space of a tile, the substitution
	of diagonal edges leaves $G_\rightarrow$ bounded from below by~$3(5cd-2)^2$.
	Suppose the described approach constructs $75c^2 d$ many flows.
	This implies that the algorithm encountered a state with $\lambda = 75c^2 d \cdot \nicefrac{d}{3} - 1$ existing sets $S_1,\dots, S_\lambda$ and the current path~$P_i$ had to be assigned to a new set $S_{\lambda+1}$.
	We conclude that each set $S_1, \dots, S_\lambda$ contained a path $P_j$ that included the first edge $e_1$ of $P_i$, meaning that there were at least ${75c^2 d\cdot d = 3(5cd)^2 > 3(5cd -2)^2}$ many paths that contained $e_1$, which contradicts~our~premise.
\end{proof}

\subsubsection{Phase~\num{3.3}: Interior flow realization}
\label{sec:subflow_realization}

In order to exchange robots between tile interiors as modeled by the flow $G_{\scaffold_5}$, we have to determine a collision-free protocol that allows robots to pass through the scaffold and into adjacent tiles.
To this end, we describe a set of movement patterns for the robots of a single tile;
these realize a single $d$-subflow in a stable manner within $\mathcal{O}(d)$ steps.
We then derive a compact concatenation of schedules that realize a $d$-subflow each, finally realizing up to $d$ such flows by a schedule of makespan $\mathcal{O}(d)$.

\paragraph*{An invariant interior configuration}
For the definition of invariant configurations, we consider a decomposition of $T$ into \emph{layers}, where the $i$th layer is the boundary of $T$ after iteratively deleting the layers $1,\ldots,i-1$.
Thus, the first layer is directly adjacent to the scaffold.

We say that the configuration of $T\in\scaffold_5$ is \emph{push-stable} (with respect to a flow $G_{\scaffold_5}$), if the following criteria are met for every robot $r$ on an $i$th layer of the interior of $T$.

\begin{description}
	\item[$i \leq d$:] $r$ is connected to the closest boundary robot by a straight line of robots.
	\item[$i > d$:] $r$ is connected to a robot $r'$ on layer $d$ by a path of
	robots in higher-order layers than~$d$, such that the closest side to $r'$ of the
	boundary of $T$ rests on an edge without outgoing flow.  
\end{description}

\begin{figure}[h]
	\centering
	\begin{subfigure}[b]{0.33\linewidth}
	\centering
	\def\svgwidth{0.9\columnwidth}
	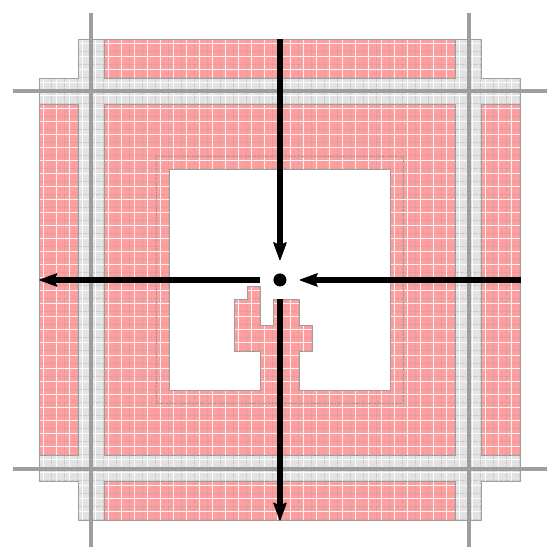
	\caption{}
	\label{fig:hull_creation-a}
\end{subfigure}\hfil
\begin{subfigure}[b]{0.33\linewidth}
	\centering
	\def\svgwidth{0.9\columnwidth}
	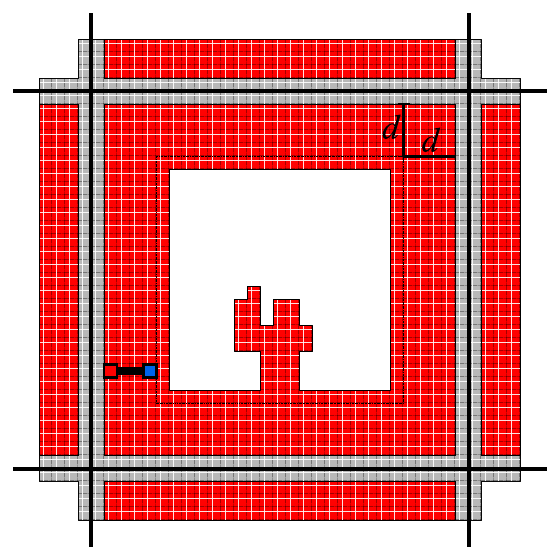
	\caption{}
	\label{fig:hull_creation-b}
\end{subfigure}\hfil
\begin{subfigure}[b]{0.33\linewidth}
	\centering
	\def\svgwidth{0.9\columnwidth}
	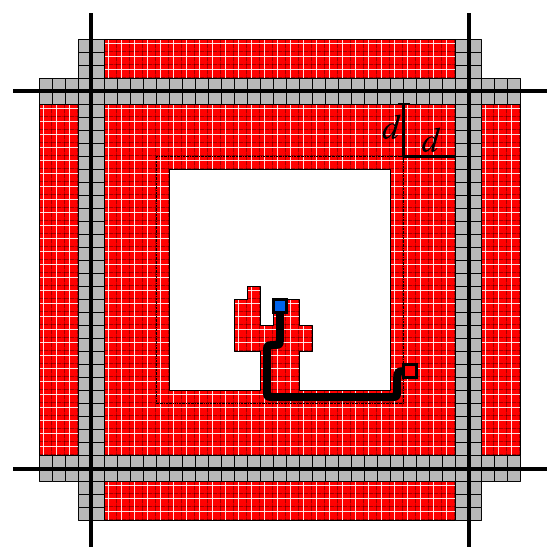
	\caption{}
	\label{fig:hull_creation-c}
\end{subfigure}\hfil
	\caption{(a) An example of a push-stable configuration of a tile with two incoming and two outgoing edges. 
	(b) Highlighted are a robot on layer $d$ and (c) on a higher-order layer, with paths that ensure their connectivity.}
	\label{fig:hull_creation}
\end{figure}

A \emph{total sink} (\emph{total source}) is a tile $T$ that has four incoming (outgoing) edges of non-zero value over $G_{\scaffold_5}$.
Conversely, a \emph{partial sink} (\emph{partial source}) is a tile $T$ that is not flow-conserving over
$G_{\scaffold_5}$, but has no more than three incoming (outgoing) edges of non-zero value.
Note that by definition, total sources can never be configured in a push-stable manner.
As this creates the need to handle total sinks separately, we provide a separate approach later on in this section.
Note that, although this constraint does not apply to total sources, we can apply the outlined methods for total sinks in reverse, and therefore refer to the same section.

\paragraph*{Realizing a single subflow}
We now provide the detailed description of our approach to realize a single subflow.
The approach consists of the following three subphases.

\begin{description}
	\item[Phase~3.3.1:] Interior preprocessing.
	\item[Phase~3.3.2:] Matching incoming and outgoing robots.
	\item[Phase~3.3.3:] Pushing robots into their target tiles.
\end{description}

\begin{lemma}
	Consider a $d$-subflow~$H_{\scaffold_5}\subseteq G_{\scaffold_5}$ and a tile $T\in\scaffold_5$ that is flow-conserving with respect to $H_{\scaffold_5}$.
	There is a stable schedule of makespan $\mathcal{O}(d)$ that realizes the flow at~its~location.
\end{lemma}

Using~\Cref{the:local_od}, \emph{Phase~3.3.1} constructs a push-stable initial configuration that places $0<d'\leq d$ outgoing robots adjacent to their respective target tile's shared boundary in the interior of $T$.

Subsequently, a local application of~\Cref{the:rotatesort} is used to embed those outgoing robots into the boundary between their current and target tile, allowing them to be pushed into the interior of a neighboring tile in a single step, see~\Cref{fig:paths_on_hulls-c}.

For simplified notation, parts of the scaffold are temporarily ``reassigned'' to neighboring tiles for this purpose.
For instance, a tile $T$ may temporarily surrender ownership of its half of the scaffold separating it from a
neighboring tile $T'$, see~\Cref{fig:paths_on_hulls-b}, assuming~${\mathit{f}'((T', T)) > 0}$.
The $2\times 2$ corner areas are never reassigned.

In \emph{Phase~3.3.2}, we perform a pairwise matching of incoming and outgoing robots, so that each pair can be connected by a crossing-free path.
Each such path passes through a unique layer $d+i$ for some $i\in\mathbb{N}$, depending on the matching,
meaning that some paths might pass through incomplete or missing layers (see~\Cref{fig:paths_on_hulls-c}).
Note that any such path begins adjacent to a robot that enters $T$, travels inwards to the $(d+i)$th layer,
at which point follows said layer until it travels outwards again to the boundary of $T$, reaching the vertex at which the matched robot leaves $T$.
This approach closely follows the methods described by Demaine et al.~\cite{dfk+-cmprs-19} in their paper on a related reconfiguration problem.

\begin{figure}[h]
	\centering
\begin{subfigure}[b]{0.24\linewidth}
	\centering
	\def\svgwidth{\columnwidth}
	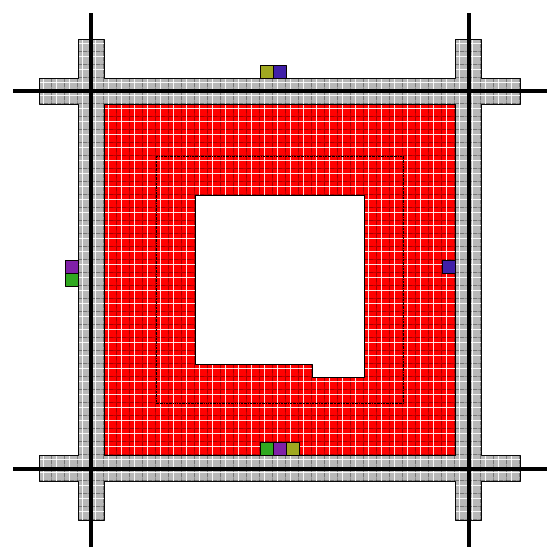
	\caption{}
	\label{fig:paths_on_hulls-a}
\end{subfigure}\hfil
\begin{subfigure}[b]{0.24\linewidth}
	\centering
	\def\svgwidth{\columnwidth}
	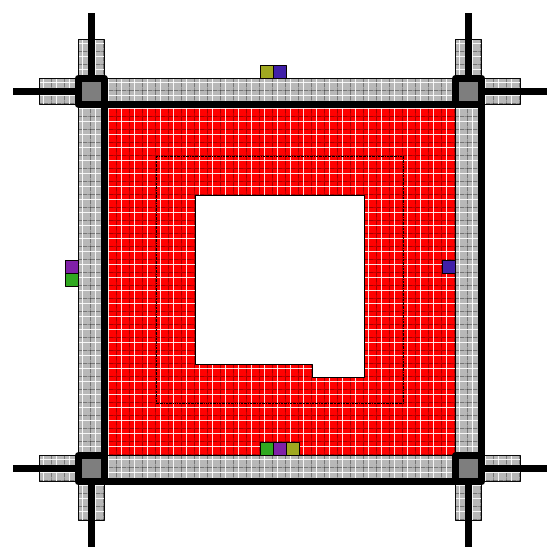
	\caption{}
	\label{fig:paths_on_hulls-b}
\end{subfigure}\hfil
\begin{subfigure}[b]{0.24\linewidth}
	\centering
	\def\svgwidth{\columnwidth}
	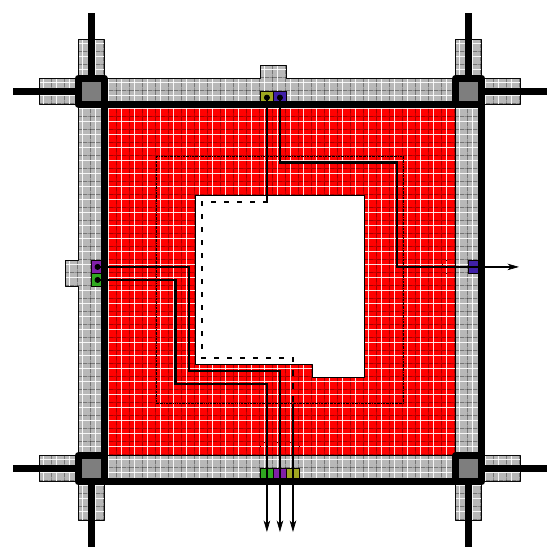
	\caption{}
	\label{fig:paths_on_hulls-c}
\end{subfigure}\hfil
\begin{subfigure}[b]{0.24\linewidth}
	\centering
	\def\svgwidth{\columnwidth}
	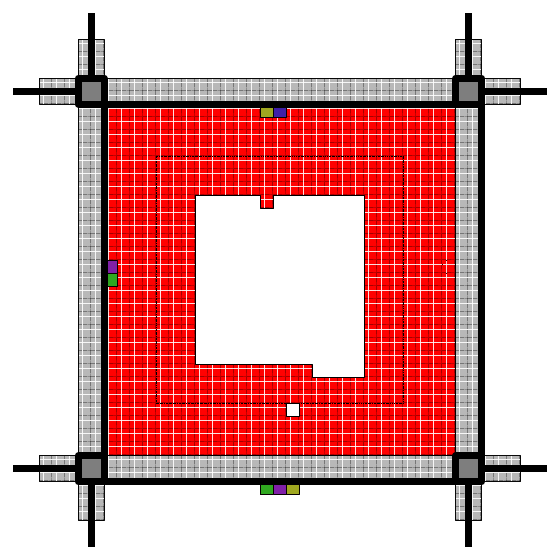
	\caption{}
	\label{fig:paths_on_hulls-d}
\end{subfigure}\hfil
	\caption{A tile $T\in\scaffold_5$ during the
		realization of a single $d$-subflow. The dashed portions of paths indicate that no motion actually
		occurs in this segment due to the described pushing behavior.}
	\label{fig:paths_on_hulls}
\end{figure}

The actual realization of $H_{\scaffold_5}$ at tile $T$ can now be performed in a single, simultaneous pushing motion along each of the constructed paths in \emph{Phase~3.3.3}, see~\Cref{fig:paths_on_hulls-d}.

\begin{lemma}
	The realization of a $d$-subflow can be performed in a way that results in another push-stable configuration of $T$.
\end{lemma}

\begin{proof}
	The crucial steps for this occur in Phase~3.3.3.
	Consider a path that connects an incoming and an outgoing robot.
	If the path is fully occupied, the described pushing motion does not have any effect on the silhouette of the tile it passes through, so this motion by itself yields a push-stable configuration.
	
	Now consider a path through an incomplete layer that has an incoming robot entering it at the first position.
	We continue the resulting pushing motion until the first empty position on the path, which then becomes occupied.
	
	The tail end of the path is handled differently.
	The last $d$ positions on the path may be occupied by robots that are then always connected directly to the scaffold through one another.
	We push only these (at most) $d$ robots towards the boundary, potentially leaving a hole at the $d$th layer of the tile, see~\Cref{fig:paths_on_hulls-d}.
	Note that the positions immediately before the last $d$ may be occupied by robots as well.
	However, any robot above layer $d$ must be connected to an \emph{incoming} side of the boundary via a path through its own layer.
	This ensures that any robot above the $d$th layer remains connected regardless of the outgoing robots moving away
	from it.
	
	We conclude that 
	Phase~3.3.3 always results in a push-stable configuration.
\end{proof}

We slightly modify the approach described before, in order to apply it to non-flow-conserving tiles.

\begin{lemma}
	Consider a $d$-subflow~$H_{\scaffold_5}\subseteq G_{\scaffold_5}$ and a tile $T\in\scaffold_5$ that is not a total sink or total source with respect to $G_{\scaffold_5}$.
	There is a stable schedule of makespan $\mathcal{O}(d)$ that realizes the flow at its location and yields another push-stable configuration.
	\label{lem:partial_sink_source_construction}
\end{lemma}
\begin{proof}
	Consider a tile $T$ that is not flow-conserving with respect to a $d$-subflow $H_{\scaffold_5}$ and for which there are $\lambda \leq d$ robots either entering or leaving the
	tile that cannot be matched to a robot doing the opposite. See~\Cref{fig:partial_source_and_sink} for an illustration.
	
	If $T$ is a partial source, a push-stable initial configuration may be created by the previously outlined method.
	As the pulling motion at the outgoing end of a path always results in another push-stable configuration, the lack of incoming robots does not cause any disconnections.
	
	If $T$ is a partial sink, a slight adjustment to the initial push-stable configuration must be made.
	As $\lambda \leq d$, it is possible to leave $\lambda$ empty positions in the $d$th layer, behind outgoing robots.
	The result is a push-stable configuration that allows the incoming paths to push into the empty space behind the
	outgoing robots, reaching another push-stable configuration.
\end{proof}
\begin{figure}[h]
	\centering
	\begin{subfigure}[b]{0.24\linewidth}
		\centering
		\def\svgwidth{\columnwidth}
		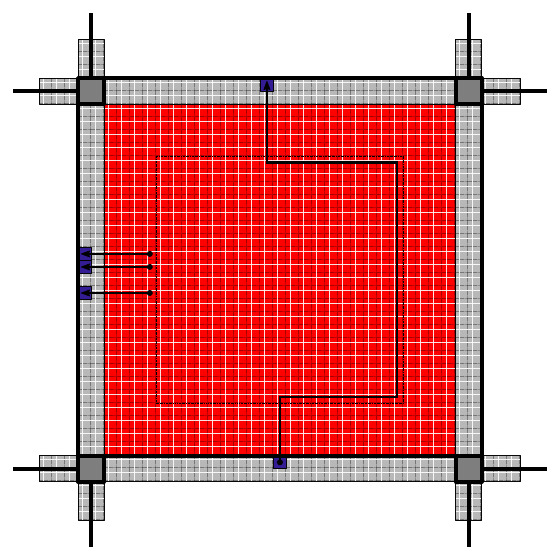
		\caption{}
		\label{fig:partial_source_and_sink-a}
	\end{subfigure}\hfil
	\begin{subfigure}[b]{0.24\linewidth}
		\centering
		\def\svgwidth{\columnwidth}
		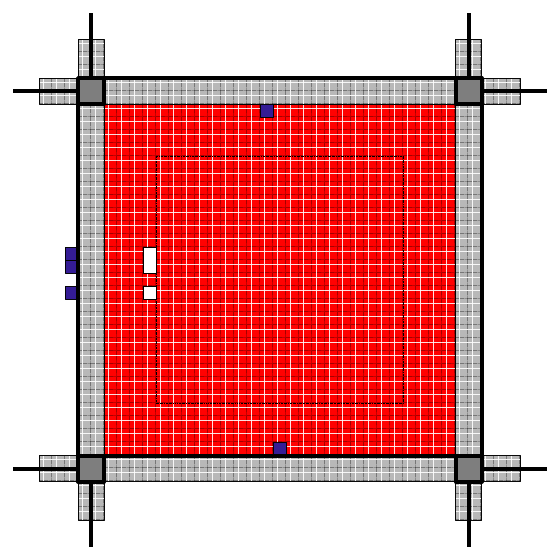
		\caption{}
		\label{fig:partial_source_and_sink-b}
	\end{subfigure}\hfil
	\begin{subfigure}[b]{0.24\linewidth}
		\centering
		\def\svgwidth{\columnwidth}
		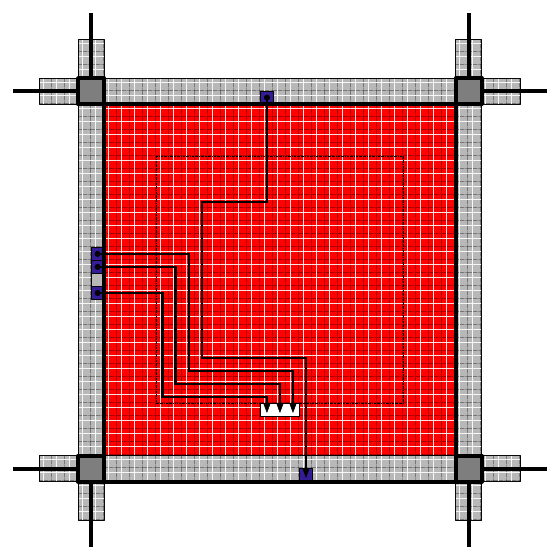
		\caption{}
		\label{fig:partial_source_and_sink-c}
	\end{subfigure}\hfil
	\begin{subfigure}[b]{0.24\linewidth}
		\centering
		\def\svgwidth{\columnwidth}
		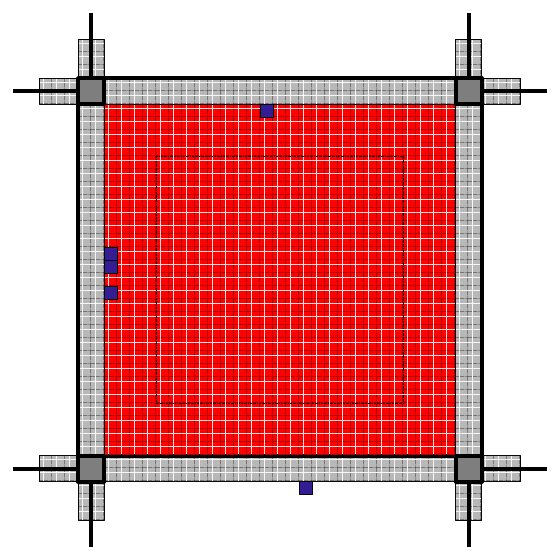
		\caption{}
		\label{fig:partial_source_and_sink-d}
	\end{subfigure}\hfil
	\caption{(a+b) A partial source tile and (c+d) a partial sink tile during the realization of a single $d$-subflow}
	\label{fig:partial_source_and_sink}
\end{figure}

\paragraph*{Realizing multiple subflows}

The previously described movement patterns may be compactly combined into a single schedule of makespan $\mathcal{O}(d)$
that realizes up to $d$ many $d$-subflows at once.

\begin{lemma}
	Consider a tile $T\in\scaffold_5$ that is not a total sink or total source with respect to~$G_{\scaffold_5}$ and a sequence $S_{\scaffold_5} := (H^1_{\scaffold_5}, \dots, H^\ell_{\scaffold_5} )$ of $d$-subflows of
	$G_{\scaffold_5}$ with $\ell \leq d$.
	There exists a stable schedule of makespan $\mathcal{O}(d)$ that realizes all $\ell$ many $d$-subflows.
	\label{lem:partial_d_subflows}
\end{lemma}
\begin{proof}
	The first portion of this schedule consists of an application of~\Cref{the:local_od} that prepares for the final $\ell$ transformation steps that realize one subflow each, see~\Cref{fig:multiple_paths_on_hulls}.
	This is achieved by stacking the agents against the target tile's boundary in successive layers, for which the order of agents is based on the order of subflows in $S_{\scaffold_5}$, with later subflows occupying locations on higher-order layers.
	As $\ell \leq d$, no queue may be more than $d$ agents long.
	The result is a push-stable configuration of~$T$, which takes into account the special cases discussed in~\Cref{lem:partial_sink_source_construction}.
	
	Next, we apply~\Cref{the:local_od} to embed the queues of agents in the wall between their current and their target tile, placing corresponding scaffold agents at the very back of each queue, see~\Cref{fig:multiple_paths_on_hulls-a}.
	This leaves only the actual realization steps to be determined.
	
	These consist of $\ell$ successive pushing operations based on matchings between incoming and outgoing agents, ensuring that each intermittent configuration of the schedule is stable.
\end{proof}

\begin{figure}[h]
	\centering
	\begin{subfigure}[b]{0.24\linewidth}
		\centering
		\def\svgwidth{\columnwidth}
		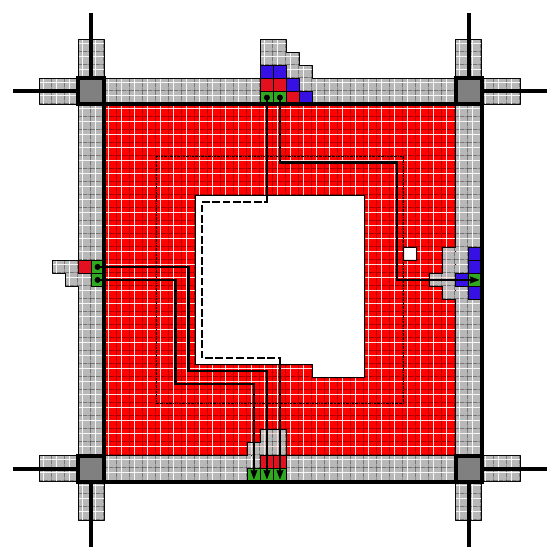
		\caption{}
		\label{fig:multiple_paths_on_hulls-a}
	\end{subfigure}\hfil
	\begin{subfigure}[b]{0.24\linewidth}
		\centering
		\def\svgwidth{\columnwidth}
		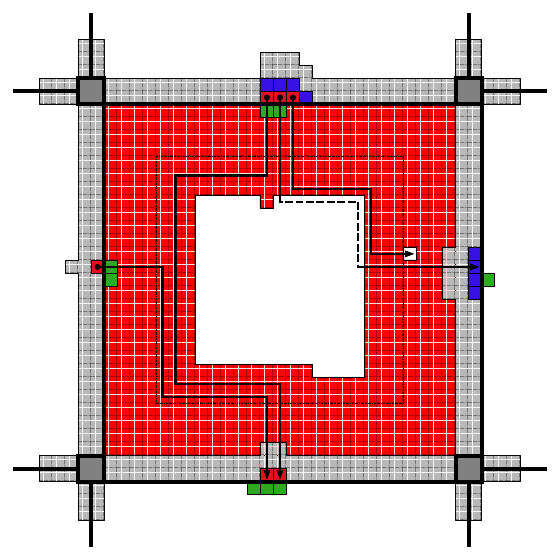
		\caption{}
		\label{fig:multiple_paths_on_hulls-b}
	\end{subfigure}\hfil
	\begin{subfigure}[b]{0.24\linewidth}
		\centering
		\def\svgwidth{\columnwidth}
		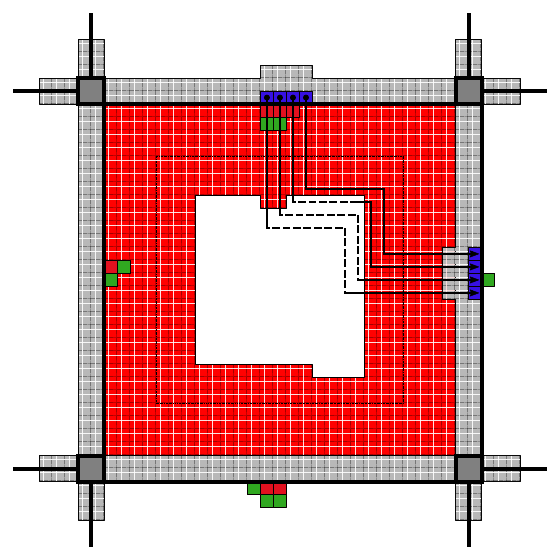
		\caption{}
		\label{fig:multiple_paths_on_hulls-c}
	\end{subfigure}\hfil
	\begin{subfigure}[b]{0.24\linewidth}
		\centering
		\def\svgwidth{\columnwidth}
		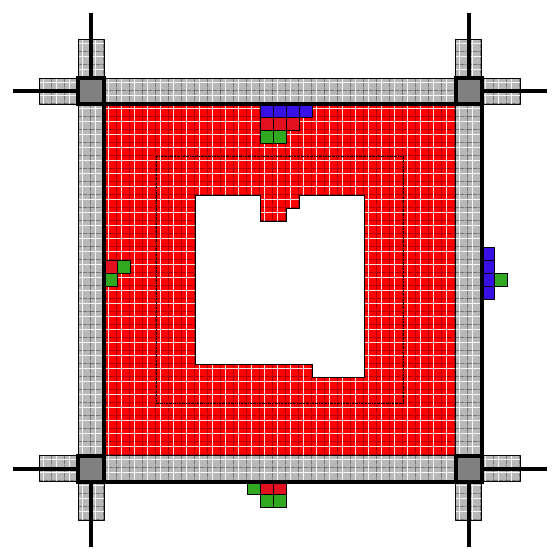
		\caption{}
		\label{fig:multiple_paths_on_hulls-d}
	\end{subfigure}\hfil
	\caption{A tile $T\in\scaffold_5$ during the realization of $\ell$ $d$-subflows.}
	\label{fig:multiple_paths_on_hulls}
\end{figure}

\paragraph*{Total sinks and sources}
As noted before, total sinks and sources form a special case that we handle separately from the described push-stable patterns.

\begin{lemma}
	Consider a tile $T\in\scaffold_5$ that is a total sink or total source with respect to $G_{\scaffold_5}$ and a sequence $S_{\scaffold_5} := (H^1_{\scaffold_5}, \dots, H^\ell_{\scaffold_5} )$ of $d$-subflows of
	$G_{\scaffold_5}$ with $\ell \leq d$.
	There exists a stable schedule of makespan~$\mathcal{O}(d)$ that realizes all $\ell$ many $d$-subflows.
	\label{lem:total_d_subflows}
\end{lemma}
\begin{proof}
	We divide the interior space of $T$ into four disjoint triangles along the diagonals of the tile, see~\Cref{fig:total_sink_filling}.
	Without loss of generality, assume that $T$ is a total sink.
	No more than $d^2$ robots can enter at any given edge as a result of $S_{\scaffold_5}$, and the resulting number of robots in the tile's interior cannot be larger than $(5cd-2)^2$.
	This means that it is sufficient to leave the corresponding number of positions at the innermost section of each triangle unoccupied.
	The motion of entering robots cannot cause any issues with stability and takes no more than $\ell$ transformations.
	This implies a total makespan of $\mathcal{O}(d) + \ell$.
	
	As every schedule is invertible, the opposite movement pattern may be applied to total sources;
	pushing outwards from the very center of the tile while only moving robots within a single triangle region does not cause any disconnections of the configuration.
\end{proof}

\begin{figure}[h]
	\centering

	\begin{subfigure}[b]{0.19\linewidth}
		\centering
		\def\svgwidth{\columnwidth}
		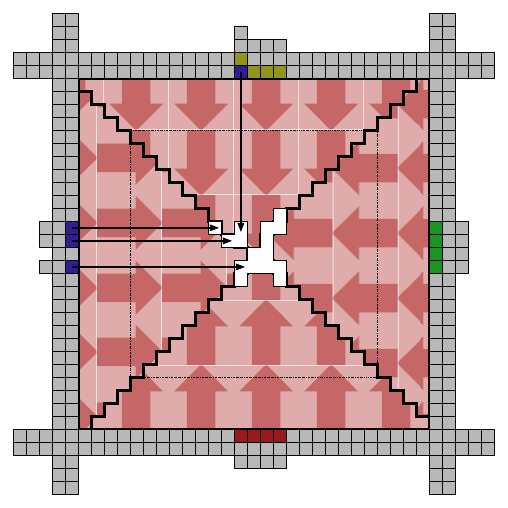
		\caption{}
		\label{fig:total_sink_filling-a}
	\end{subfigure}\hfil
	\begin{subfigure}[b]{0.19\linewidth}
		\centering
		\def\svgwidth{\columnwidth}
		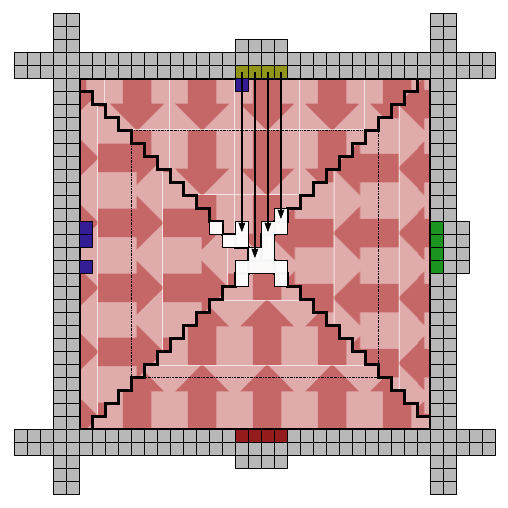
		\caption{}
		\label{fig:total_sink_filling-b}
	\end{subfigure}\hfil
	\begin{subfigure}[b]{0.19\linewidth}
		\centering
		\def\svgwidth{\columnwidth}
		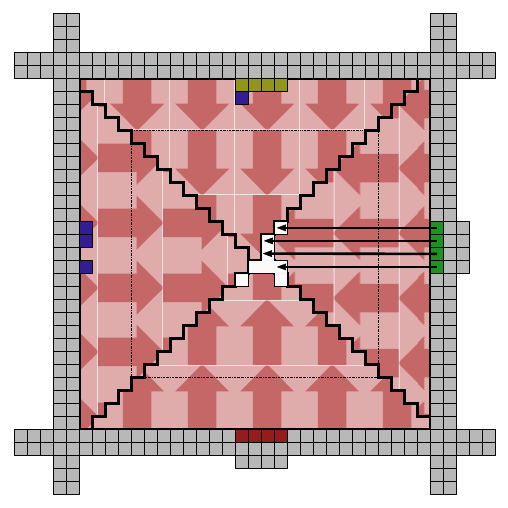
		\caption{}
		\label{fig:total_sink_filling-c}
	\end{subfigure}\hfil
	\begin{subfigure}[b]{0.19\linewidth}
		\centering
		\def\svgwidth{\columnwidth}
		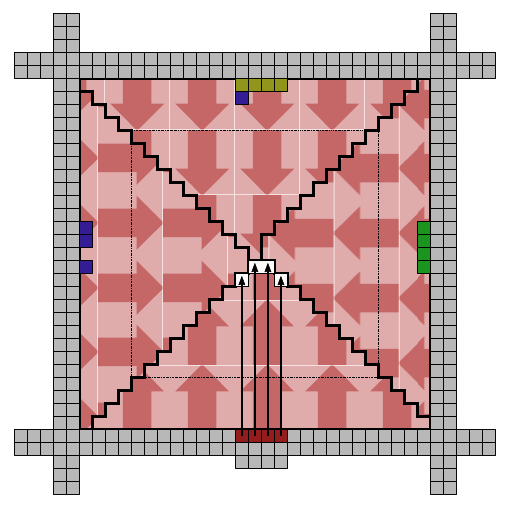
		\caption{}
		\label{fig:total_sink_filling-d}
	\end{subfigure}\hfil
	\begin{subfigure}[b]{0.19\linewidth}
		\centering
		\def\svgwidth{\columnwidth}
		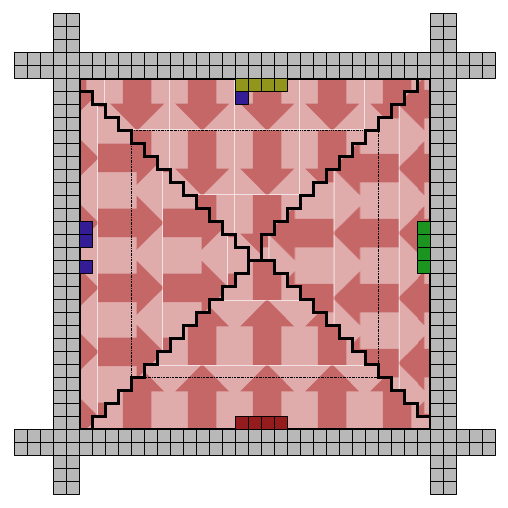
		\caption{}
		\label{fig:total_sink_filling-e}
	\end{subfigure}\hfil
	\caption{A total sink $T\in\scaffold_5$ during the
		realization of $\ell$ many $d$-subflows. As indicated by the triangular marks, each side has
		a unique space into which robots may push.}
	\label{fig:total_sink_filling}
\end{figure}

\paragraph*{Realizing all subflows}

By applying~\Cref{lem:partial_d_subflows,lem:total_d_subflows}, up to $d$ many $d$-subflows of
$G_{\scaffold_5}$ may be realized in~$\mathcal{O}(d)$ transformations.
As we created $\mathcal{O}(d)$ such subflows in~\Cref{subsec:flow_partitioning}, all of them may be realized through $\nicefrac{\mathcal{O}(d)}{d} = \mathcal{O}(1)$ repetitions.
These repetitions require a total of $\mathcal{O}(d)$~transformations.

\subsection{Phase~\num{4}: Boundary flow}
\label{sec:boundary_reordering}

The next phase of our algorithm deals with the movement of robots that are part of the scaffold.
As described in \Cref{sec:scaffold_construction}, the robots forming the scaffold in a tiled configuration must remain in the $1$-neighborhood of their target tile over $\scaffold_5$.
The intermediate mapping step in the construction process guarantees that this remains the case even after that phase concludes.
However, the scaffold does not necessarily consist of the same robots in the tiled configurations $C_s'$ and $C_t'$.
For any given tile $T\in\scaffold_5$, up to $20cd-4$ robots exist in its neighborhood $N_1[T]$ that need to become part of its boundary structure.

This observation allows us to model the necessary exchange as a $(20cd-4)$-flow.
Using the techniques discussed in \Cref{subsec:flow_partitioning} with some minor modifications, we obtain and realize a $(d,\mathcal{O}(1))$-partition of this new flow graph $G_{\scaffold_5}^s$ as follows.

\paragraph*{Creating a planar and unidirectional flow}
Just as with the prior flow $G_{\scaffold_5}$, bidirectional and crossing edges need to be removed from the flow before the partitioning process.
This can be achieved in the same manner as before.
By applying~\Cref{the:local_od}, we place the robots that need to be exchanged in the wall between adjacent tiles, and subsequently exchange them by means of local rotations.
As at most~$20cd-4$ robots need to be swapped between any two adjacent tiles, this process takes~$\mathcal{O}(d)$ steps in total.
Once this process has been applied to all tiles, $G_{\scaffold_5}^s$~is bounded from above by $3(20cd-4)$.

\paragraph*{Partitioning the flow}
We again seek to decompose the flow into smaller components that can be realized separately.
A~$(d,\mathcal{O}(1))$-partition can be computed by the following method.
By applying~\Cref{lem:flow_component_separation}, we compute a $(3(20cd-4), 2)$-partition of $G_{\scaffold_5}^s$, decomposing the flow into acyclic and cyclic components $G_\rightarrow^s$ and $G_\bigcirc^s$.
Using modified versions of the partitioning algorithms from~\Cref{subsec:flow_partitioning}, we then compute a $(d,\mathcal{O}(1))$-partition of each of them.

By simply modifying the final step of the algorithm used to obtain~\Cref{lem:circulation_partitioning}, we can compute a  $(d,\mathcal{O}(1))$-partition of $G_\bigcirc^s$.
Thus, we obtain the following.
\smallskip
\begin{lemma}
	We can compute a $(d,\mathcal{O}(1))$-partition of a circulation $G_\bigcirc^s$ which is a $(k\cdot d)$-flow for some constant $k\in\mathbb{N}$ in polynomial time.
\end{lemma}
\begin{proof}
	We refer to~\Cref{lem:circulation_partitioning} for the first two steps of the algorithm.
	Then, let $\mathbb{P}_X^j\in \{\mathbb{P}_\circlearrowright^1, \mathbb{P}_\circlearrowright^2,
	\mathbb{P}_\circlearrowleft^1, \mathbb{P}_\circlearrowleft^2\}$ refer to a partition constructed in the second step.
	Additionally, let ${F = (\mathbb{P}^j_X, E_F)}$ refer to the forest induced by the nesting properties of cycles within the partitions.
	
	By labeling the vertices of this forest by their depth in $F \bmod k$, it is possible to construct~$\mathcal{O}(1)$ flows $\{G_{\bigcirc_1}^s, G_{\bigcirc_2}^s, \dots \}$, each being the union of cycles sharing a~label.

	It remains to show that each flow $G_{\bigcirc_i}^s$ with $i\in[1,k]$ is a $d$-subflow.
	For this, consider any two cycles $u,v\in \mathbb{P}_X^j$ that share a common edge $e$.
	By construction, one of the two cycles must lie within the other.
	Without loss of generality, assume that $v$ lies in $u$.
	This implies that there exists a path from $v$ to its root via $u$ in $F$, with all cycles on the path between
	$v$ and $u$ sharing the edge $e$ as well.
	As all edges in $G_\bigcirc^s$ are bounded from above by $k\cdot d$, the cycles containing $e$ lie on a path of length no more than $k\cdot d$ in $F$.
	Thus, $e$ has a weight of at most $\nicefrac{k\cdot d}{k} = d$ in every $G_{\bigcirc_i}^s$.
	We conclude that each flow $G_{\bigcirc_i}^s$ is a $d$-subflow of $G_\bigcirc^s$.
\end{proof}

We now adjust~\Cref{lem:supplydemand_partitioning} to compute a $(d,\mathcal{O}(1))$-partition of $G_\rightarrow^s$.
\begin{lemma}
	We can compute a $(d,\mathcal{O}(1))$-partition of an acyclic flow  $G_\rightarrow^s$ which is a $(k\cdot d)$-flow for some constant $k\in\mathbb{N}$ in polynomial time.
\end{lemma}
\begin{proof}
	See the proof of \Cref{lem:supplydemand_partitioning} for the algorithm itself.
	We know that $G_\rightarrow^s$ is bounded from above by $k\cdot d$.
	Suppose that the algorithm constructs $k+1$ flows, implying that the algorithm encountered a configuration in which a path $P_i$ could not be assigned to the existing $\lambda = 3k \cdot \nicefrac{d}{3}$ sets $S_1, \dots, S_\lambda$, spawning a new set $S_{\lambda+1}$.
	This implies that each set $S_1, \dots, S_\lambda$ contained a path $P_j$ that shares the first edge $e_1$ of~$P_i$, meaning that there were at least $k\cdot d +1 > k\cdot d$ paths containing $e_1$, implying
	a flow value of more than $k\cdot d$.
\end{proof}

\paragraph*{Realizing a single subflow}
We now argue that the computed subflows can be realized in $\mathcal{O}(d)$ transformations, and begin with the realization of a single subflow.

\begin{lemma}
	Let $H_s = (\scaffold_5, E_{\scaffold_5}, \mathit{f}_s')$ be a $d$-subflow of $G_{\scaffold_5}^s$.
	We can efficiently compute a stable schedule of makespan $\mathcal{O}(d)$ that realizes $H_s$.
\end{lemma}
\begin{proof}
	The exact movement comes down to the number of robots within each tile the flow passes through.
	
	Consider any tile $T\in\scaffold_5$ that has non-zero flow in $H_s$ and let $\eta\leq(5cd)^2$ be the initial number of robots within the tile (including the scaffold).
	Additionally, let $\eta_\text{out}\leq d$ denote the exact number of robots that need to leave $T$ to realize $H_s$ at its location.
	Each tile always contains at least its boundary robots, so $\eta\geq 20cd-4$.
	The existing movement patterns from~\Cref{sec:subflow_realization} may be utilized to realize the flow at $T$'s location if $\eta-\eta_\text{out}\geq 20cd-4 $.
	
	On the other hand, if $\eta-\eta_\text{out} < 20cd-4$, there are not sufficiently many robots within the tile's interior to close up the
	holes in the scaffold through which the $\eta_\text{out}$ robots leave it as part of the pushing motion.
	Note that this can only occur at flow-conserving tiles, because source tiles cannot end up with less robots than
	$20cd-4$.
	This means that there are $\eta_\text{in} = \eta_\text{out}$ robots entering $T$ as part of the same flow.
	
	As every outgoing side of $T$ is adjacent to an incoming side of another tile $T'$, no part of the scaffold can
	become disconnected through the pushing motion.
	Subsequently, there are at most $\eta_\text{in}$ unoccupied positions in the boundary of $T$.
	By pushing along and into the boundary from the middle of the line of $\eta_\text{in}$ robots that were pushed into~$T$,
	those robots may patch all holes in $\mathcal{O}(\eta_\text{in}) =~\mathcal{O}(d)$ transformations, as depicted in~\Cref{fig:scaffold_stability}.
	Overall, this realizes a single $d$-subflow of $G_{\scaffold_5}^s$ in $\mathcal{O}(d)$ steps.
\end{proof}

\begin{figure}[h]
	\centering
	\def\svgwidth{\columnwidth}
	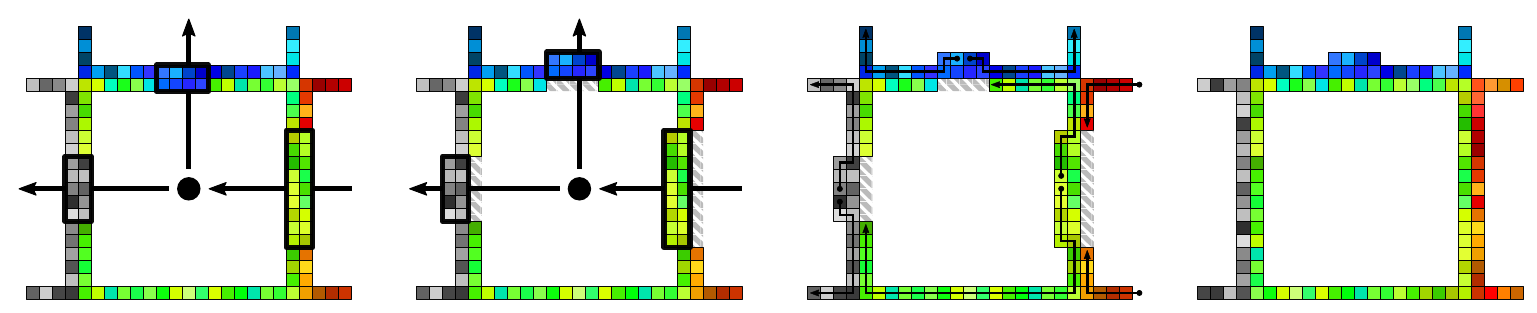
	\caption{We can employ the displayed movement patterns to exchange boundary robots between adjacent tiles.}
	\label{fig:scaffold_stability}
\end{figure}

\paragraph*{Realizing all subflows}
It is now easy to see that all subflows can be realized in a total of $\mathcal{O}(d)$ transformations; thus, we obtain the following.

\begin{lemma}
	For any pair $(C_s,C_t)$ with diameter $d$, we can compute a stable schedule with makespan $\mathcal{O}(d)$ that realizes $G_{s}$, transforming a labeled tiled
	configuration $C_s'$ into another labeled tiled configuration $C_t'$, in polynomial time.
\end{lemma}
\begin{proof}
	The entirety of $G_{\scaffold_5}^s$ may be realized by applying the previous approach to each of the $d$-subflows, bringing all scaffold robots into their correct target tiles.
	As there are $\mathcal{O}(1)$ such $d$-subflows, the total makespan of this  is $\mathcal{O}(d)$.
\end{proof}

\subsection{Phases~\num{5} and~\num{6}: Scaffold deconstruction}
\label{subsec:scaffold-deconstruction}

Once Phase~\num{4} concludes, we have reached a tiled configuration in which every tile contains precisely the robots that it would in a tiled configuration $C_t'$ of the target configuration.
This means that we can reconfigure into $C_t'$ by a single application of~\Cref{the:local_od}, which forms the entirety of Phase~\num{5}.

As Phase~\num{6} is a reverse of Phase~\num{2}, this concludes the description of the algorithm.
Because each phase takes $\mathcal{O}(d)$~transformation steps, this proves~\Cref{the:c_lower}.

\section{Conclusions and further considerations}
\label{sec:further-considerations}
We have provided new results for efficient coordinated motion planning for a labeled swarm of robots.
In particular, we resolved two major open problems for connected reconfiguration:
in the labeled case, a stretch factor of $\Omega(\sqrt{n})$ may be inevitable for $n$ robots, and constant stretch can be achieved for scaled arrangements of labeled robots;
previously, the former was unknown for any kind of connected reconfiguration, while the latter was only known for the (considerably easier) unlabeled case.

These results pose a number of relevant follow-up questions and provide insights into related problems.
In the following, we briefly discuss both direct implications and further research questiona.

\subsection{Colored reconfiguration}
\label{subsec:colored}

This paper focuses on \emph{labeled} robots, so each robot has a known start and end
position; on the other hand, 
Fekete et al.~\cite{connected-motion-journal} considered the unlabeled version of the problem,
in which robots can freely be exchanged to achieve a desired final shape.
A natural generalization of both is
\emph{colored} reconfiguration,
in which we have $\ell\in \mathbb{N}$ different classes of objects that are exchangeable within
each class.
In this setting, the unlabeled case corresponds to $\ell=1$, while the labeled case amounts to $\ell=N$.

By employing a minimum bottleneck matching between objects in the respective
color classes to assign final positions from the initial configuration,
before applying our labeled reconfiguration, it is straightforward to see that all our results also 
hold for the colored version of the problem.

\subsection{Achievable makespans}
\label{subsec:achievable-makespans}
An interesting open problem arises from considering lower bounds on stretch factors.
Although we provided a lower bound of $\Omega(\sqrt{n})$ on the achievable
stretch factor in some scenarios, we believe that constant stretch schedules
exist in most instances.
More specifically, we conjecture that the constant $c^*$ introduced in~\Cref{the:c_lower} is as low as~$2$, i.e., constant stretch
is achievable for any pair of configurations of scale $2$.
The thin counterexample with scale~$1$ in~\Cref{the:sqrt-n-stretch} is heavily restricted by the global implications of local movement, which appear to lose relevance quickly with growing scale.
Moreover, it is unknown whether $\mathcal{O}(\sqrt{n})$ stretch 
can always be achieved, in particular for thin configurations.

\subsection{Decentralized methods}
\label{subsec:distributed}
Another important set of questions arises from running the ensuing protocols in
a largely decentralized fashion.
This involves two aspects of parallelization, distinguishing between
distributed methods for (1) carrying out the computations and (2) performing
the actual motion control. 

We are confident that both issues
can be addressed to some extent by making use of the hierarchical
structure of our approach, which operates on (I) the macroscopic
tile graph, as well as on (II) the set of robots within tiles. Making
use of the canonical algorithmic structures and reconfigurations within tiles,
it should be relatively straightforward to do this at the individual tile levels (II)
based on local operations (for computation) and protocols (for motion control).
It is plausible that the macrocscopic aspects (I) can be addressed by 
making use of distributed methods for computing network flows, such
as~\cite{jiang2013parallel}; using these in our context would require embedding the
high-level flow graph into the scaffold framework and implementing the
corresponding distributed algorithms as protocols into our lower-level scaffold
structures. 

However, many of the involved details appear to be quite intricate
and go beyond the main purpose of this paper.
Consequently, these and other aspects are left for future work.

\bibliography{references}

\end{document}